\documentclass{llncs}
\usepackage{pgf,tikz}
\usetikzlibrary {arrows}
\usetikzlibrary {decorations}
\usetikzlibrary{snakes}
\usepackage{quant_games}
\usepackage{amsmath,amssymb,stmaryrd,wasysym}

\usepackage[margin=3.65cm]{geometry}

\makeatletter
\let\c@definition\c@theorem
\let\c@lemma\c@theorem
\let\c@corollary\c@theorem
\let\c@remark\c@theorem
\let\c@example\c@theorem
\let\c@proposition\c@theorem

\author{Thomas~Brihaye \and V\'eronique~Bruy\`ere \and Julie~De~Pril}

\usepackage{url}
\urldef{\mail}\path|%
       {thomas.brihaye, veronique.bruyere, julie.depril}@umons.ac.be|

 \institute{University of Mons - UMONS\\
   Place du Parc 20, 7000 Mons, Belgium\\
   \mail} 

 \title{On Equilibria in Quantitative Games with Reachability/Safety
   Objectives}

\begin{document}
\maketitle

\begin{abstract}
  In this paper, we study turn-based quantitative multiplayer non
  zero-sum games played on finite graphs with both reachability and
  safety objectives. In this framework a player with a reachability
  objective aims at reaching his own goal as soon as possible, whereas
  a player with a safety objective aims at avoiding his bad set or, if
  impossible, delaying its visit as long as possible.  We prove the
  existence of Nash equilibria with finite memory in quantitative
  multiplayer reachability/safety games. Moreover, we prove the
  existence of finite-memory secure equilibria for quantitative
  two-player reachability games.
\end{abstract}


\section{Introduction}

\paragraph{General framework.} 
The construction of correct and efficient computer systems (hardware
or software) is recognized as an extremely difficult task. To support
the design and verification of such systems, mathematical
logic, automata theory~\cite{HU} and
more recently model-checking~\cite{CGP00} have been intensively
studied. The model-checking approach, which is now an important part
of the design cycle in industries, has proved its efficiency when
applied to systems that can be accurately modeled as a finite-state
automaton. In contrast, the application of these techniques to
computer software, complex systems like embedded systems or
distributed systems has been less successful. This could be partly
explained by the following reasons: classical automata-based models do
not faithfully capture the complex interactive behavior of modern
computational systems that are usually composed of several interacting
components, also interacting with an environment that is only
partially under control. Recent research works show that it is
suitable to generalize automata models used in the classical approach
to verification, with the more flexible and mathematically deeper
game-theoretic framework~\cite{Na50,OR94}.

\paragraph{Game theory meets automata theory.} 
The basic framework that extends computational models with concepts
from game theory is the so-called two-player zero-sum games played on
graphs~\cite{GTW02}. Many problems in verification and design of
reactive systems can be modeled with this approach, like modeling
controller-environment interactions.  Given a model of a system
interacting with a hostile environment, given a control objective
(like preventing the system to reach some bad configurations), the
controller synthesis problem asks to build a controller ensuring that
the control objective is enforced whatever the environment will
do. Two-player zero-sum games played on graphs are adequate models to
solve this problem~\cite{T95}. Moves of Player 1 model actions of the
controller whereas moves of Player 2 model the uncontrollable actions
of the environment, and a winning strategy for Player 1 is an abstract
form of a control program that enforces the control objective.

The controller synthesis problem is suitable to model purely
antagonist interactions between a controller and a hostile
environment. However in order to study more complex systems with more
than two components whose objectives are not necessarily antagonist,
we need multiplayer and non zero-sum games to model them
adequately. Moreover, we do not look for winning strategies, but
rather try to find relevant notions of equilibria, for instance the
famous notion of Nash equilibria~\cite{Na50}. On the other hand, only
qualitative objectives have been considered so far to specify, for
example, that a player must be able to reach a target set of states in
the underlying game graph. But, in line with the previous point, we
also want to express and solve games for quantitative objectives such
as forcing the game to reach a particular set of states within a given
time bound, or within a given energy consumption limit. In summary, we
need to study \emph{equilibria} for \emph{multiplayer non zero-sum}
games played on graphs with \emph{quantitative} objectives. This
article provides some new results in this research direction.

\paragraph{Related work.}

Several recent papers have considered two-player zero-sum games played
on finite graphs with regular objectives enriched by some
\emph{quantitative aspects}. Let us mention some of them: games with
\emph{finitary objectives}~\cite{CH06}, games with \emph{prioritized
  requirements}~\cite{AKW08}, \emph{request-response games} where the
waiting times between the requests and the responses are
minimized~\cite{HTW08,Z09}, and games whose winning conditions are
expressed via \emph{quantitative languages}~\cite{BCHJ09}.

Other works concern qualitative non zero-sum games.
The notion of secure equilibrium, an interesting refinement of Nash
equilibrium, has been introduced in~\cite{CJH06}. It has been proved
that a unique secure equilibrium always exists for two-player non zero-sum games
with regular objectives. In~\cite{GU08}, general criteria ensuring
existence of Nash equilibria, subgame perfect equilibria (resp. secure
equilibria) are provided for $n$-player (resp. $2$-player)
games, as well as complexity results.

  
Finally, we mention reference \cite{BG09} that combines both
quantitative and non zero-sum aspects. It is maybe the nearest related
work compared to us, however the framework and the objectives are
pretty different.  In~\cite{BG09}, the authors study games played on
graphs with terminal vertices where quantitative payoffs are assigned
to the players. These games may have cycles but all the infinite plays
form a single outcome (like in chess where every infinite play is a
draw). That paper gives criteria that ensure the existence of Nash
(and subgame perfect) equilibria in pure and memoryless strategies.


\paragraph{Our contribution.} 

We here study turn-based quantitative multiplayer non zero-sum games
played on finite graphs with reachability objectives. In this
framework each player aims at reaching his own goal as soon as
possible. We focus on existence results for two solution concepts:
Nash equilibrium and secure equilibrium.  We prove the existence of
finite-memory Nash (resp. secure) equilibria in $n$-player
(resp. $2$-player) games.  Moreover, we prove that given a Nash
(resp. secure) equilibrium of a $n$-player (resp. $2$-player) game,
we can build a finite-memory Nash (resp. secure) equilibrium of the
\emph{same type}, i.e. preserving the set of players achieving their
objectives.  For the case of Nash equilibria, we extend our results in
two directions. First we prove that finite-memory Nash equilibria
still exist when the model is enriched by allowing $n$-tuples of
non-negative costs on edges (one cost by player). This result provides
an answer to a question we posed in~\cite{BBD10}. Secondly, we prove
the existence of Nash equilibria in quantitative games where both
safety and reachability objectives coexist.

Our results are not a direct consequence of the existing results in
the qualitative framework, they require some new proof techniques.  To
the best of our knowledge, this is the first general result about the
existence of equilibria in quantitative multiplayer games played on
graphs.

\paragraph{Organization of the paper.} 
Section~\ref{sec:prelim} is dedicated to definitions. We present the
games and the equilibria we study. In Section~\ref{sec:NE} we first
prove an existence result for Nash equilibria and provide the
finite-memory characterization.  Similar results concerning secure
equilibria in two-player games are established in Section~\ref{sec:SE}.
Finally, in Section~\ref{sec:ext}, we discuss the extensions of our
results on Nash equilibria.

A part of these results has been published in~\cite{BBD10}, namely the
existence of finite-memory Nash (resp. secure) equilibria in
multiplayer (resp. $2$-player) games, and the fact that given a Nash
equilibrium we can build a finite-memory Nash equilibrium of the same
type. Additionally in this paper we give proofs of the previous
results and we extend our existence result for Nash equilibria in the
two directions mentioned above, namely $(i)$ $n$-tuples of
non-negative costs on edges and $(ii)$ reachability/safety
objectives. Moreover, in the two-player case, we prove that given a
secure equilibrium, we can build a finite-memory secure equilibrium of
the same type.

\section{Preliminaries}\label{sec:prelim}

\subsection{Definitions}\label{sub:def}

We consider here \emph{quantitative} games played on a graph where all
the players have \emph{reachability\footnote{The general case of
    reachability/safety objectives is handled in
    Subsection~\ref{sub:safety}.}  objectives}.  It means that, given
a certain set of vertices~$\F_i$, each player~$i$ wants to reach one
of these vertices as soon as possible.

This section is mainly inspired by reference \cite{GU08}.

\begin{definition}\label{def-game}
  An \emph{infinite turn-based quantitative multiplayer reachability
    game} is a tuple $\game$ where
  \begin{itemize}
  \item[\textbullet] $\Pi$ is a finite set of players,
  \item[\textbullet] $\G$ is a finite directed graph where $V$ is the
    set of vertices, $(V_i)_{i \in \Pi}$ is a partition of $V$ into
    the state sets of each player, $v_0 \in V$ is the initial vertex,
    and $E \subseteq V \times V$ is the set of edges, and
  \item[\textbullet] $\F_i \subseteq V$ is the goal set of player~$i$.
  \end{itemize}
\end{definition}

We assume that each vertex has at least one outgoing edge.  The game
is played as follows. A token is first placed on the vertex~$v_0$.
Player~$i$, such that $v_0 \in V_i$, has to choose one of the outgoing
edges of~$v_0$ and put the token on the vertex~$v_1$ reached when
following this edge. Then, it is the turn of the player who
owns~$v_1$. And so on.

A \emph{play}~$\rho \in V^{\omega}$ (resp. a \emph{history}~$h \in
V^+$) of $\mathcal{G}$ is an \emph{infinite} (resp. a \emph{finite})
path through the graph~$G$ starting from vertex~$v_0$. Note that a
history is always non empty because it starts with~$v_0$.  The set~$H
\subseteq V^+$ is made up of all the histories of~$\mathcal{G}$.  A
\emph{prefix} (resp.  \emph{proper prefix})~$\la$ of a
history~$h=h_0\ldots h_k$ is a finite sequence $h_0\ldots h_l$, with
$l\leq k$ (resp. $l<k$), denoted by $\la \leq h$ (resp. $\la < h$).
We similarly consider a prefix~$\la$ of a play~$\rho$, denoted by~$\la
< \rho$.

We say that a play~$\rho=\rho_0\rho_1\ldots$ \emph{visits} a set~$S
\subseteq V$ (resp. a vertex~$v \in V$) if there exists $l \in \N$
such that $\rho_l$ is in $S$ (resp. $\rho_l=v$). The same terminology
also stands for a history~$h$.  Similarly, we say that $\rho$
\emph{visits}~$S$ \emph{after} (resp. \emph{in}) \emph{a
  prefix~$\rho_0\ldots\rho_k$} if there exists $l > k$ (resp.  $l \leq
k$) such that $\rho_l$ is in $S$.  For any play~$\rho$ we denote by
$\visit(\rho)$ the set of players~$i \in \Pi$ such that $\rho$ visits
$\F_i$. The set~$\visit(h)$ for a history~$h$ is defined similarly.
The function~$\last$ returns, given a history~$h= h_0\ldots h_k$, the
last vertex~$h_k$ of $h$, and the \emph{length}~$|h|$ of~$h$ is the
number~$k$ of its \emph{edges}\footnote{Note that the length is not
  defined as the number of vertices.}.

For any play~$\rho=\rho_0\rho_1\ldots$ of $\mathcal{G}$, we note
$\payoff_i(\rho)$ the \emph{cost} of player~$i$, defined by:
\[ \payoff_i(\rho) = \left\{
\begin{array}{ll}
  l & \mbox{ if $l$ is the \emph{least} index such that $\rho_l \in
    \F_i$,}\\ + \infty & \mbox{ otherwise.}
\end{array}\right. \]
We note $\payoff(\rho) = (\payoff_i(\rho))_{\ipi}$ the \emph{cost
  profile} for the play~$\rho$.  The aim of each player~$i$ is to
\emph{minimize} the cost he has to pay, i.e. reach his goal set~$\F_i$
as soon as possible.

A \emph{strategy} of player~$i$ in $\mathcal{G}$ is a
function~$\sigma: V^* V_i \to V$ assigning to each history~$hv$ ending
in a vertex~$v$ of player~$i$, a next vertex~$\sigma(hv)$ such that
$(v,\sigma(hv))$ belongs to $E$. We say that a
play~$\rho=\rho_0\rho_1\ldots$ of $\mathcal{G}$ is \emph{consistent}
with a strategy~$\sigma$ of player~$i$ if
$\rho_{k+1}=\sigma(\rho_0\ldots\rho_k)$ for all $k \in \N$ such that
$\rho_k \in V_i$. The same terminology is used for a history~$h$ of
$\mathcal{G}$.  A \emph{strategy profile} of $\mathcal{G}$ is a
tuple~$(\sigma_i)_{\ipi}$ where $\sigma_i$ is a strategy for
player~$i$.  It determines a unique play of $\mathcal{G}$ consistent
with each strategy~$\sigma_i$, called the \emph{outcome} of
$(\sigma_i)_{\ipi}$ and denoted by $\langle (\sigma_i)_{\ipi}
\rangle$.

A strategy~$\sigma$ of player~$i$ is \emph{\memoryless} if $\sigma$
depends only on the current vertex, i.e. $\sigma(hv)=\sigma(v)$ for
all~$h \in H$ and~$v \in V_i$. More generally, $\sigma$ is a
\emph{finite-memory strategy} if the equivalence
relation~$\approx_{\sigma}$ on $H$ defined by $h \approx_{\sigma} h'$
if $\sigma(h\delta) = \sigma(h'\delta)$ for all~$\delta \in V^*V_i$
has finite index. In other words, a finite-memory strategy is a
strategy that can be implemented by a finite automaton with output.  A
strategy profile~$(\sigma_i)_{\ipi}$ is called \emph{\memoryless} or
\emph{finite-memory} if each~$\sigma_i$ is a \memoryless\ or a
finite-memory strategy, respectively.

For a strategy profile~$(\sigma_i)_{\ipi}$ with outcome $\rho$ and a
strategy~$\sigma_j'$ of player~$j$ ($j\in \Pi$), we say that
\emph{player~$j$ deviates from~$\rho$ after a prefix~$h$ of~$\rho$} if
there exists a prefix~$h'$ of~$\rho$ such that $h \leq h'$, $h'$ is
consistent with $\sigma_j'$ and $\sigma_j'(h') \not = \sigma_j(h')$.
We also say that \emph{player~$j$ deviates from~$\rho$ just after a
  prefix~$h$ of $\rho$} if $h$ is consistent with $\sigma_j'$ and
$\sigma_j'(h) \not = \sigma_j(h)$.


We now introduce the notion of \emph{Nash equilibrium} and
\emph{secure equilibrium}.
\begin{definition}\label{def:ne}
  A strategy profile~$(\sigma_i)_{\ipi}$ of a game~$\mathcal{G}$
  is a \emph{Nash equilibrium} if for all player~$j \in \Pi$ and for all
  strategy~$\sigma_j'$ of player~$j$, we have:
  $$\payoff_j(\rho) \leq \payoff_j(\rho')$$ where $\rho = \langle
  (\sigma_i)_{\ipi} \rangle$ and $\rho' = \langle
  \sigma_j',(\sigma_i)_{\ipimj} \rangle$.
\end{definition}
This definition means that player~$j$ (for all~$j \in \Pi$) has no
incentive to deviate since he increases his cost when
using~$\sigma_j'$ instead of~$\sigma_j$.  Keeping notations of
Definition~\ref{def:ne} in mind, a strategy~$\sigma_j'$ such that
$\payoff_j(\rho) > \payoff_j(\rho')$ is called a \emph{profitable
  deviation} for player~$j$ with respect to~$(\sigma_i)_{\ipi}$. In
this case either player~$j$ pays an infinite cost for~$\rho$ and a
finite cost for $\rho'$ ($\rho'$ visits~$\F_j$, but $\rho$ does not),
or player~$j$ pays a finite cost for~$\rho$ and a strictly lower cost
for~$\rho'$ ($\rho'$ visits~$\F_j$ earlier than~$\rho$ does).


As our results on secure equilibria stand for two-player games, we
define this notion only in this context. In order to define the
concept of secure equilibrium\footnote{Our definition naturally
  extends the notion of \emph{secure equilibrium} proposed
  in~\cite{CJH06} to the quantitative reachability framework. A longer
  discussion comparing the two notions can be found in
  Section~\ref{sec:qualiquanti}.}  we first need to associate two
appropriate binary relations~$\prec_1$ and~$\prec_2$ on cost profiles
with player~1 and~2 respectively.  Given two cost profiles~$(x_1,x_2)$
and $(y_1,y_2)$:
\begin{align*}
  & (x_1,x_2) \prec_1 (y_1,y_2) \quad \text{iff} \quad (x_1 > y_1)
  \vee (x_1=y_1 \wedge x_2 < y_2)\,.
\end{align*}
We then say that \emph{player~$1$ prefers} $(y_1,y_2)$ \emph{to}
$(x_1,x_2)$. In other words, player~$1$ prefers a cost profile to
another either if he can decrease his own cost, or if he can increase
the cost of player~$2$, while keeping his own cost.  We define the
relation~$\prec_2$ symmetrically.

\begin{definition}\label{def:ES}
  A strategy profile~$(\sigma_1,\sigma_2)$ of a two-player
  game~$\mathcal{G}$ is a \emph{secure equi\-li\-brium} if there does
  not exist any strategy~$\sigma_1'$ of player~$1$ such that:
  $$\payoff(\rho) \prec_1 \payoff(\rho')$$ where $\rho = \langle
  \sigma_1,\sigma_2 \rangle$ and $\rho' = \langle \sigma_1',\sigma_2
  \rangle$, and there does not exist any strategy~$\sigma_2'$ of
  player~$2$ such that:
  $$\payoff(\rho) \prec_2 \payoff(\rho')$$ where $\rho = \langle
  \sigma_1,\sigma_2 \rangle$ and $\rho' = \langle \sigma_1
  ,\sigma_2'\rangle$.
\end{definition}
In other words, player~$1$ (resp. $2$) has no incentive to deviate,
with respect to the relation~$\prec_1$ (resp. $\prec_2$).  Note that
any secure equilibrium is a Nash equi\-li\-brium.  A
stra\-te\-gy~$\sigma_j'$ such that $\payoff(\rho) \prec_j
\payoff(\rho')$ is called a \emph{$\prec_j$-profitable deviation} for
player~$j$ with respect to $(\sigma_1,\sigma_2)$ (for $j \in
\{1,2\}$).

Let us go back to the multiplayer framework and define the notion of
\type\ of an equilibrium.
\begin{definition}
  The \type\ of a strategy profile~$(\sigma_i)_{\ipi}$ in a
  rea\-cha\-bi\-li\-ty game~$\mathcal{G}$ is the set of players~$j \in
  \Pi$ such that the outcome~$\rho$ of $(\sigma_i)_{\ipi}$ visits
  $\F_j$. It is denoted by $\valtype((\sigma_i)_{\ipi})$.
\end{definition}
In other words, $\valtype((\sigma_i)_{\ipi})=\visit(\rho)$.


The previous definitions are illustrated in the following example.


\begin{example}\label{ex:EN-ES}
  Let $\mathcal{G}=(V,V_1,V_2,v_0,E,\F_1,\F_2)$ be the two-player game
  depicted in
  Figure~\ref{fig-ex}. The states of player~$1$ (resp. $2$) are
  represented by circles (resp. squares)\footnote{We will keep this convention
  through the article.}. Thus, according to Figure~\ref{fig-ex}, $V_1 =
  \{A,C,D\}$ and $V_2=\{B\}$, the initial vertex~$v_0$ is the vertex~$A$,
  and we set $\F_1=\{C\}$ and $\F_2=\{D\}$.
 
  \begin{figure}[h!]
    \centering
    \begin{tikzpicture}[->,>=stealth',shorten >=1pt,auto,node
        distance=3cm,bend angle=20]
      \everymath{\scriptstyle}
      \tikzstyle{j0}=[draw,circle,text centered,,inner sep=2pt]
      \tikzstyle{j1}=[draw,rectangle,text centered,,inner sep=4pt]
      \tikzstyle{j2}=[draw,circle,fill=black!20!white,text centered,,inner 
sep=2pt]
      \tikzstyle{j3}=[draw,circle,double,text centered,,inner sep=2pt]
      
      \node[j0]  (0) {$A$};
      \node[j1]  (1) [left of=0] {$B$};
      \node[j3]  (2) [right of=0] {$D$};
      \node[j2]  (3) [below of=1,xshift=1.5cm,yshift=1cm] {$C$};
      \node (4) [above of=0,yshift=-2.25cm] {};
      \path
      (4) edge (0)
      (0) edge[bend right]  (1)
      (1) edge[bend right]  (0)
      (0) edge[bend left]  (2)
      (2) edge[bend left] (0)
      (1) edge  (3)
      (3) edge  (0);

    \end{tikzpicture} \caption{A two-player game with
    $\F_1=\{C\}$ and  $\F_2=\{D\}$.}  
\label{fig-ex} 
\end{figure}
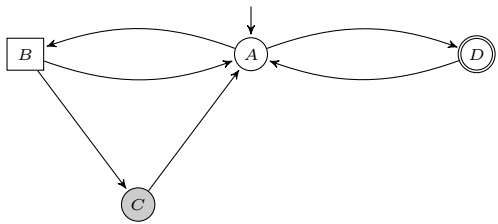
 
    An example of play in~$\mathcal{G}$ is given by $\rho =
    (AD)^\omega$, which visits~$\F_2$ but not~$\F_1$, leading to the
    cost profile $\payoff((AD)^\omega) = (+ \infty,1)$. The
    play~$\rho$ is, among others, the outcome of the
    strategy\footnote{Note that player~1 has no choice in vertices~$C$
      and~$D$, that is, $\sigma_1(hv)$ is necessarily equal to~$A$
      for~$v \in \{C,D\}$.}  profile~$(\sigma_1,\sigma_2)$ where
    $\sigma_1(hA)=D$ and $\sigma_2(hB)=C$, for all histories~$h$.

    Let us show that the strategy profile~$(\sigma_1,\sigma_2)$ 
    is not a Nash equilibrium,
    by proving that player~$1$ has a profitable deviation~$\sigma_1'$
    in which he
    manages to decrease his own cost. With $\sigma_1'$ defined by
    $\sigma_1'(hA)=B$, we get the play $\langle
    \sigma_1',\sigma_2 \rangle = (ABC)^\omega$ such that
    $\payoff((ABC)^\omega)=(2,+\infty)$, and in particular
    $\payoff_1((ABC)^\omega)< \payoff_1(\rho)$.  
                
    On the opposite side, one can show that $(\sigma_1',\sigma_2)$ is
    a Nash equilibrium.  However $(\sigma_1',\sigma_2)$ is not a
    secure equilibrium. Indeed, player~$2$ has a $\prec_2$-profitable
    deviation in which he can increase player~$1$'s cost without
    modifying his own cost.  With~$\sigma_2'$ the strategy of
    player~$2$ defined by $\sigma_2'(hB)=A$, we get the play $\langle
    \sigma_1',\sigma_2' \rangle = (AB)^\omega$ such that
    $\payoff((AB)^\omega) = (+\infty,+\infty)$, and $\payoff(\langle
    \sigma_1',\sigma_2 \rangle) \prec_2 \payoff(\langle
    \sigma_1',\sigma_2' \rangle)$.

    Notice that all strategies
    discussed so far are \memoryless. In order to
    obtain a Nash equilibrium of \type\ $\{1,2\}$, finite-memory strategies
    are necessary. We define the following finite-memory strategy
    profile~$(\tau_1,\tau_2)$:
     \begin{align*} \tau_1(hA) = \begin{cases} D
    & \text{if } h=\epsilon\\ B & \text{if }
    h\ne\epsilon \end{cases} \quad ; \quad \tau_2(hB)= \begin{cases} C
    & \text{if } h \text{ visits } D\\ A
    & \text{otherwise.}  \end{cases} \end{align*} 
     The outcome~$\pi=\langle (\tau_1,\tau_2) \rangle$ is 
     equal to~$AD(ABC)^\omega$
    and has costs $(4,1)$. In order to prove that $(\tau_1,\tau_2)$
    is a Nash equilibrium, we prove that no player has a
    profitable deviation. 
    For player~2 it is clearly impossible to get a cost less than~1. 
    To try to get a cost less than~4, player~1 must use a 
    strategy~$\tau_1'$ such that~$\tau_1'(A)=B$. But then player~2 
    chooses~$\tau_2(AB)=A$. The prefix~$ABA$ of the outcome 
    of~$(\tau_1',\tau_2)$
    shows that player~1 will increase his cost of 4.

    However $(\tau_1,\tau_2)$ is not a secure
    equilibrium since player~$2$ has a $\prec_2$-profitable
    deviation~$\tau_2'$ such that $\tau_2'(hB)=A$ for all histories~$h$.
     One can show that, in this example, there is no
    secure equilibrium of \type\ $\{1,2\}$. 
\end{example}

The questions studied in this article are the following ones:
\begin{pbm}\label{pbm 1}
  Given~$\mathcal{G}$ a quantitative multiplayer (resp. two-player)
  reachability game, does there exist a Nash equilibrium (resp. a
  secure equilibrium) in~$\mathcal{G}$?
\end{pbm}

\begin{pbm}\label{pbm 2}
  Given a Nash equilibrium (resp. a secure equilibrium) in a
  quantitative multiplayer (resp. two-player) reachability
  game~$\mathcal{G}$, does there exist a finite-memory Nash
  equilibrium (resp. secure equilibrium) with the same \type?
\end{pbm}
We provide positive answers in Sections~\ref{sec:NE}
and~\ref{sec:SE}. Notice that these problems have been investigated in
the qualitative framework (see~\cite{GU08}).

\subsection{Qualitative Games vs Quantitative Games}\label{sec:qualiquanti}
We show in this section that Problems~\ref{pbm 1} and~\ref{pbm 2} can
not be reduced to problems on qualitative games.


Given a quantitative multiplayer reachability game~$\mathcal{G}$, one
can naturally define a \emph{qualitative} version of~$\mathcal{G}$,
denoted by~$\qualiG$, such that the payoffs\footnote{For qualitative
games, we use the notion of \emph{payoff} rather than the notion
of \emph{cost} since \win \ (resp. \lose) can be seen as a payoff of
$1$ (resp. $0$) and the aim of the players is to maximize their
payoffs.} are \emph{qualitative}. Given a play~$\rho$
of~$\mathcal{G}$, the qualitative payoff of player~$i$ is defined by:
\[ \payoffquali_i(\rho) = \left\{
\begin{array}{ll}
  \win & \mbox{ if } \payoff_i(\rho) \mbox{ is finite}\\        
  \lose & \mbox{ otherwise.}
\end{array}\right. \]

We note~$\payoffquali(\rho) = (\payoffquali_i(\rho))_{\ipi}$ the
\emph{qualitative} payoff profile for the play~$\rho$. In this framework,
player~$i$ aims at reaching his own goal set, i.e. at obtaining
payoff~$\win$. With this idea in mind, one can naturally adapt the notion
of Nash (resp. secure) equilibrium to the qualitative framework. 

The existence of Nash (resp. secure) equilibria 
in~$n$-player (resp.~$2$-player) qualitative games~$\qualiG$ has been
proved in~\cite[Corollary~12]{GU08} (resp. \cite[Theorem~2]{CJH06})
for reachability objectives, and more generally for Borel
objectives.

The next example illustrates that lifting Nash equilibria 
in~$\qualiG$ to Nash equilibria in~$\mathcal{G}$ does not work.
We developed new ideas in Sections~\ref{sec:NE} and~\ref{sec:SE}
to solve Problem~\ref{pbm 1}.

\begin{example}\label{ex:qualiquanti}
  Let us now consider the two-player game~$\mathcal{G}$ depicted in
  Figure~\ref{fig-ex-sequaliquanti}, such that~$\F_1=\{B,E\}$
  and~$\F_2=\{C\}$. Notice that only player~1 effectively plays in
  this game. We are going to exhibit a secure (and thus Nash)
  equi\-li\-brium~$(\sigma_1,\sigma_2)$ in the qualitative
  game~$\qualiG$ that can not be lifted neither to a secure nor to a
  Nash equilibrium in the quantitative game~$\mathcal{G}$. The
  strategy profile~$(\sigma_1,\sigma_2)$ is defined such that~$\langle
  (\sigma_1,\sigma_2) \rangle = ADE^{\omega}$. It is a secure
  equilibrium in~$\qualiG$ with the qualitative payoff
  profile~$(\win,\lose)$.  However~$(\sigma_1,\sigma_2)$ is not a Nash
  (and thus not a secure) equilibrium in~$\mathcal{G}$. Indeed, the
  play~$ABC^{\omega}$ provides a smaller cost to player~1,
  i.e.~$\payoff_1(ABC^{\omega}) < \payoff_1(ADE^{\omega})$. Notice
  that in this example, there is no equilibrium in~$\mathcal{G}$ of
  type~$\{1\}$.

 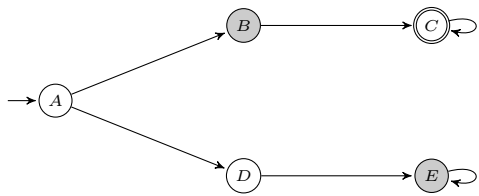
\begin{figure}[h!]
   \centering
     \begin{tikzpicture}[->,>=stealth',shorten >=1pt,auto,node
         distance=2.5cm,bend angle=20]
       \everymath{\scriptstyle}
       \tikzstyle{j0}=[draw,circle,text centered,,inner sep=2pt]
       \tikzstyle{j1}=[draw,rectangle,text centered,,inner sep=4pt]
       \tikzstyle{j2}=[draw,circle,double,text centered,,inner sep=2pt]
       \tikzstyle{j3}=[draw,circle,fill=black!20!white,text centered,,inner 
sep=2pt]

       \node[j0]  (0) {$A$};
       \node[j3]  (1) [right of=0,yshift=1cm] {$B$};
       \node[j0]  (2) [right of=0,yshift=-1cm] {$D$};
       \node[j2]  (3) [right of=1] {$C$};
       \node[j3]  (4) [right of=2] {$E$};

       \node (5) [left of=0,xshift=1.75cm] {};
       \path
       (5) edge (0)
       (0) edge  (1)
       (1) edge  (3)
       (0) edge  (2)
       (2) edge (4)
       (3) edge[loop right]  (3)
       (4) edge[loop right]  (4);

     \end{tikzpicture}
     \caption{A game~$\mathcal{G}$ with an equilibrium in~$\qualiG$ that
     can not be lifted to $\mathcal{G}$.}
     \label{fig-ex-sequaliquanti}
\end{figure}

\end{example}

The next proposition shows that on the opposite side, any Nash
equilibrium in a quantitative game~$\mathcal{G}$ can be lifted
to a Nash equilibrium in the qualitative game~$\qualiG$.

\begin{proposition}\label{prop:qualiquanti NE}
  If~$(\sigma_i)_{\ipi}$ is a Nash equilibrium in a quantitative
  multiplayer reachability game~$\mathcal{G}$, 
  then~$(\sigma_i)_{\ipi}$ is also a Nash
  equilibrium in~$\qualiG$.
\end{proposition}
\begin{proof}
  For a contradiction, let us assume that in~$\qualiG$, player~$j$ has
  a profitable deviation~$\sigma_j'$ w.r.t.~$(\sigma_i)_{\ipi}$.  This
  is only possible if~$\payoffquali_j(\langle
  (\sigma_i)_{\ipi}\rangle) = \lose$ and $\payoffquali_j(\langle
  \sigma_j',(\sigma_i)_{\ipimj} \rangle) = \win$. Thus when
  playing~$\sigma_j'$ against~$(\sigma_i)_{\ipimj}$, player~$j$
  manages to visit~$\F_j$. Clearly enough,~$\sigma_j'$ would also be a
  profitable deviation w.r.t. $(\sigma_i)_{\ipi}$ in~$\mathcal{G}$,
  contradicting the hypothesis.  \qed
\end{proof}

Note that
Proposition~\ref{prop:qualiquanti NE} is false for secure equilibria.
To see that, let us come back to the game~$\mathcal{G}$ of 
Figure~\ref{fig-ex-sequaliquanti}. The strategy profile~$(\sigma_1,
\sigma_2)$ such that $\langle \sigma_1,\sigma_2\rangle
 = ABC^{\omega}$ is a secure
equilibrium in the quantitative game~$\mathcal{G}$ but not in
the qualitative game~$\qualiG$.

\subsection{Unraveling}

In the proofs of this article we need to unravel the
graph~$G=(V,(V_i)_{i\in \Pi},v_0,E)$ from the initial vertex~$v_0$,
which ends up in an \emph{infinite tree}, denoted by~$\tinf$.  This
tree can be seen as a new graph where the set of vertices is the
set~$H$ of histories of~$\mathcal{G}$, the initial vertex is~$v_0$,
and a pair~$(hv,hvv') \in H \times H$ is an edge of~$\tinf$ if~$(v,v')
\in E$. A history~$h$ is a vertex of player~$i$ in~$\tinf$
if~$\last(h) \in V_i$, and it belongs to the goal set of player~$i$
if~$\last(h) \in \F_i$.

We denote by~$\ginf$ the related game. This game~$\ginf$ played on the
un\-ra\-ve\-ling~$\tinf$ of~$G$ is equivalent to the
game~$\mathcal{G}$ that is played on~$G$ in the following sense.  A
play~$(\rho_0)(\rho_0\rho_1)(\rho_0\rho_1\rho_2)\ldots$ in~$\ginf$
induces a unique play~$\rho = \rho_0\rho_1\rho_2\ldots$
in~$\mathcal{G}$, and conversely.  Thus, we denote a play in~$\ginf$
by the respective play in~$\mathcal{G}$. The bijection between plays
of~$\mathcal{G}$ and plays of~$\ginf$ allows us to use the same cost
function~$\payoff$, and to transform easily strategies
in~$\mathcal{G}$ to strategies in~$\ginf$ (and conversely).

 We also need to study the tree~$\tinf$ limited to a certain
 depth~$\depth \geq 0$: we note~$\tfin$ the \emph{truncated tree
   of~$\tinf$ of depth~$\depth$} and~$\gfin$ the \emph{finite} game
 played on~$\tfin$.  More precisely, the set of vertices of~$\tfin$ is
 the set of histories~$h \in H$ of length~$\leq\depth$; the edges
 of~$\tfin$ are defined in the same way as for~$\tinf$ except that for
 the histories~$h$ of length~$\depth$, there exists no edge~$(h,hv)$.
 A play~$\rho$ in~$\gfin$ corresponds to a history of~$\mathcal{G}$ of
 length \emph{equal to}~$\depth$.  The notions of cost and strategy
 are defined exactly like in the game~$\ginf$, but limited to the
 depth~$\depth$. For instance, a player pays an infinite cost for a
 play~$\rho$ (of length~$\depth$) if his goal set is not visited
 by~$\rho$.



\subsection{Qualitative Two-player Zero-sum Reachability Games}

In this section we recall well-known properties of
qualitative two-player zero-sum reachability games~\cite[Chapter~2]{GTW02}.
This will be necessary in our proofs.

\begin{definition}\label{def-game-2j-sn}
  A \emph{qualitative two-player zero-sum reachability game} is
  a tuple $\gamedjsn$ where 
  \begin{itemize}
  \item[\textbullet] $\Gdjsn$ is a finite directed graph 
    where~$V$ is the set of vertices,~$V_1,V_2$ is a partition of~$V$
    into the state sets of player~1 and player~2,
    and~$E \subseteq V \times V$ is the set of edges, 
  \item[\textbullet] $\F \subseteq V$ is the goal set of player~1.
  \end{itemize}
\end{definition}

Given an initial vertex~$v_0 \in V$, the notions of \emph{play}, 
\emph{history} and \emph{strategy}  are 
the same as the ones defined in Section~\ref{sub:def}.
Player~1 (resp. player~2) \emph{wins} a play~$\rho$ of~$\mathcal{G}$
if~$\rho$ visits $\F$ (resp.~$\rho$ does not visit~$\F$).
The game is said \emph{zero-sum} because every play is won by exactly one
of the two players.

In zero-sum games, it is interesting to know if one of the players can play
in such a way that he is sure to win, however the other player plays. We
can formalize this by introducing the notion of \emph{winning strategy}.
A strategy~$\sigma_i$ for player~$i$ is a winning strategy 
\emph{from an initial vertex}~$v$ if
all plays of~$\mathcal{G}$ starting in~$v$
that are consistent with~$\sigma_i$ are won by player~$i$.
If player~$i$ has a winning strategy in~$\mathcal{G}$ from~$v$, we say that
player~$i$ \emph{wins} the game~$\mathcal{G}$ from~$v$.
We say that a game~$\mathcal{G}$ is \emph{determined} if for all~$v \in V$,
one of the two players has a winning strategy from~$v$.

Martin showed \cite{M75} that every qualitative two-player zero-sum game
with a Borel type winning condition is determined. In particular, we
have the following proposition:

\begin{proposition}[\cite{GTW02}]\label{propGTW02}
  Let~$\gamedjsn$ be a qualitative two-player zero-sum reachability game.
  Then for all~$v \in V$, one of the two players has a 
  \emph{\memoryless}\ winning strategy from~$v$ (in particular,
 $\mathcal{G}$ is determined).\\
  \noindent Moreover for all vertices~$v$ from which he wins the game, player~1
  (resp. player~2)
  has a \memoryless\ strategy that is independent of~$v$ and that
  forces the play to visit~$\F$ within at most~$|V|-1$ edges (resp.
  to stay in~$V\setminus \F$).
\end{proposition}

\section{Nash Equilibria}\label{sec:NE}

From now on we will often use the term \emph{game} to denote a
quantitative multiplayer reachability game according to
Definition~\ref{def-game}.

\subsection{Existence of a Nash Equilibrium}\label{sub:ex NE}

In this section we positively solve Problem~\ref{pbm 1} for Nash
equilibria.

\begin{theorem}\label{theo:ex EN}
  In every quantitative multiplayer reachability game, there exists a
  finite-memory Nash equilibrium.
\end{theorem}

The proof of this theorem is based on the following ideas. By Kuhn's
theorem (Theorem~\ref{theo:kuhn}), there exists a Nash equilibrium in
the game~$\gfin$ played on the finite tree~$\tfin$, for any
depth~$\depth$.  By choosing an adequate depth~$\depth$,
Proposition~\ref{prop:NE fin ds inf} enables to extend this Nash
equilibrium to a Nash equilibrium in the infinite tree~$\tinf$, and
thus in~$\mathcal{G}$.  Let us detail these ideas.

We first recall Kuhn's theorem \cite{kuhn53ega}.  A \emph{preference
  relation} is a total reflexive transitive binary relation.

\begin{theorem}[Kuhn's theorem]\label{theo:kuhn}
  Let~$\tree$ be a \emph{finite} tree and~$\mathcal{G}_\tree$ a game
  played on~$\tree$. For each player~$i \in \Pi$, let~$\precsim_i$ be
  a preference relation on cost profiles.  Then there exists a
  strategy profile~$(\sigma_i)_{\ipi}$ such that for every player~$j
  \in \Pi$ and every strategy~$\sigma_j'$ of player~$j$
  in~$\mathcal{G}_\tree$ we have
  $$\payoff(\rho') \precsim_j \payoff(\rho)$$ where~$\rho=\langle
  (\sigma_i)_{\ipi} \rangle$ and~$\rho' = \langle
  \sigma_j',(\sigma_i)_{\ipimj} \rangle$.
\end{theorem}

Note that $\payoff(\rho') \precsim_j \payoff(\rho)$ means that
player~$j$ prefers the cost profile of the play~$\rho$ than the one
of~$\rho'$, or they are equivalent for him.

\begin{corollary}\label{coro:kuhn}
  Let~$\mathcal{G}$ be a game and~$\tinf$ be the unraveling
  of~$G$. Let~$\gfin$ be the game played on the truncated tree
  of~$\tinf$ of depth~$\depth$, with~$\depth \geq 0$. Then there
  exists a Nash equilibrium in~$\gfin$.
\end{corollary}

\begin{proof}
  For each player~$j \in \Pi$, we define the relation~$\precsim_j$ on
  cost profiles in the following way: let~$(x_i)_{\ipi}$
  and~$(y_i)_{\ipi}$ be two cost profiles, we say that~$(x_i)_{\ipi}
  \precsim_j (y_i)_{\ipi}$ iff~$x_j \geq y_j$. It is clearly a
  preference relation which captures the Nash equilibrium.  The
  stra\-tegy profile~$(\sigma_i)_{\ipi}$ of Kuhn's theorem is then a
  Nash equilibrium in~$\gfin$.  \qed
\end{proof}

Proposition~\ref{prop:NE fin ds inf} states that it is possible to
extend a Nash equilibrium in~$\gfin$ to a Nash equilibrium in the
game~$\ginf$, if the depth~$d$ is equal to~$(|\Pi| + 1)\cdot
2\cdot|V|$. We obtain Theorem~\ref{theo:ex EN} as a consequence of
Corollary~\ref{coro:kuhn} and Proposition~\ref{prop:NE fin ds inf}.

\begin{proposition}\label{prop:NE fin ds inf}
  Let~$\mathcal{G}$ be a game and $\tinf$ be the unraveling
  of~$G$. Let~$\gfin$ be the game played on the truncated tree
  of~$\tinf$ of depth~$\depth=(|\Pi| + 1)\cdot 2\cdot|V|$.  If there
  exists a Nash equilibrium in the game~$\gfin$, then there exists a
  finite-memory Nash equilibrium in the game~$\ginf$.
\end{proposition}


The proof of Proposition~\ref{prop:NE fin ds inf} roughly works as
follows.  Let~$(\sigma_i)_{\ipi}$ be a Nash equilibrium in~$\gfin$.  A
well-chosen prefix~$\alpha\beta$, with~$\beta$ being a cycle, is first
extracted from the outcome~$\rho$ of~$(\sigma_i)_{\ipi}$.  The outcome
of the required Nash equilibrium~$(\tau_i)_{\ipi}$ in~$\ginf$ will be
equal to~$\abo$. As soon as a player deviates from this play, all the
other players form a coalition to punish him in a way that this
deviation is not profitable for him. These ideas are detailed in
Lemmas~\ref{lemma:jeu djsn} and~\ref{lemma:NE same type}.  One can see
Lemma~\ref{lemma:jeu djsn} as a technical result used to prove
Lemma~\ref{lemma:NE same type}, which is the main ingredient to show
Proposition~\ref{prop:NE fin ds inf}. The proof of
Lemma~\ref{lemma:jeu djsn} relies on a particular case (stated below)
of Proposition~\ref{propGTW02}. More precisely, we consider the
qualitative two-player zero-sum game~$\mathcal{G}_j$ played on the
graph~$G$, where player~$j$ plays in order to reach his goal
set~$\F_j$, against the coalition of all other players that wants to
prevent him from reaching his goal set. Player~$j$ plays on the
vertices from~$V_j$ and the coalition on~$V \setminus V_j$.

\begin{proposition}[\cite{GTW02}]\label{prop:gamei}
  Let~$\gamej$ be the qualitative two-player zero-sum reachability
  game associated to player~$j$. Then player~$j$ has a \memoryless\
  stra\-tegy~$\strat_j$ that enables him to reach~$\F_j$ within
  $|V|-1$ edges from each vertex~$v$ from which he wins the
  game~$\mathcal{G}_j$.  On the contrary, the coalition has a
  \memoryless\ strategy~$\strat_{-j}$ that forces the play to stay
  in~$V \setminus \F_j$ from each vertex~$v$ from which it wins the
  game~$\mathcal{G}_j$.
\end{proposition}

\begin{lemma}\label{lemma:jeu djsn}
  Suppose~$d \geq 0$. Let~$(\sigma_i)_{\ipi}$ be a Nash equilibrium
  in~$\gfin$ and~$\rho$ the (finite) outcome of~$(\sigma_i)_{\ipi}$.
  Assume that~$\rho$ has a prefix~$\alpha\beta\gamma$, where~$\beta$
  contains at least one vertex, such that
  \begin{align*}
    & \visit(\alpha)=\visit(\alpha\beta\gamma)\\
    & \last(\alpha)=\last(\alpha\beta)\\
    & |\alpha\beta| \leq \nbband\cdot|V|\\
    & |\alpha\beta\gamma| = (\nbband+1)\cdot|V|
  \end{align*}
  for some~$\nbband \geq 1$.\\
  \noindent Let~$j \in \Pi$ be such that~$\alpha$ does not
  visit~$\F_j$.  Consider the qualitative two-player zero-sum
  game~$\gamej$.  Then for all histories~$h\vertu$ of~$\mathcal{G}$
  consistent with~$(\sigma_i)_{\ipimj}$ and such that
  $|h\vertu|\leq|\alpha\beta|$, the coalition of the players~$i
  \not=j$ wins the game~$\mathcal{G}_j$ from~$\vertu$.
\end{lemma}

Condition~$\visit(\alpha)=\visit(\alpha\beta\gamma)$ means that
if~$\F_i$ is visited by~$\alpha\beta\gamma$, it has already been
visited by~$\alpha$. Condition~$\last(\alpha)=\last(\alpha\beta)$
means that~$\beta$ is a cycle.  The play~$\rho$ of
Lemma~\ref{lemma:jeu djsn} is illustrated in Figure~\ref{fig:slicing
  play}.

Lemma~\ref{lemma:jeu djsn} says in particular that the players~$i
\not=j$ can play together to prevent player~$j$ from reaching his goal
set~$\F_j$, in case he deviates from the play~$\alpha\beta$
(as~$\alpha\beta$ is consistent with~$(\sigma_i)_{\ipimj}$).  We
denote by~$\strat_{-j}$ the \memoryless\ winning strategy of the
coalition.  For each player~$i \not= j$, let~$\strat_{i,j}$ be the
\memoryless\ strategy of player~$i$ in~$\mathcal{G}$ induced
by~$\strat_{-j}$.

\begin{proof}[of Lemma~\ref{lemma:jeu djsn}]
  By contradiction suppose that player~$j$ wins the 
  game~$\mathcal{G}_j$ from~$\vertu$. By Proposition~\ref{prop:gamei}
  player~$j$ has a \memoryless\ winning strategy~$\strat_j$
  which enables him to reach his goal set~$\F_j$ within
  at most~$|V|-1$ edges from~$\vertu$. We show that~$\strat_j$
  leads to a profitable deviation for player~$j$ w.r.t.~$(\sigma_i)_{\ipi}$ 
  in the game~$\gfin$, which is impossible
  by hypothesis. 

  Let~$\rho'$ be a play in~$\gfin$ such that~$h\vertu$ is a prefix
  of~$\rho'$, and from~$\vertu$,
  player~$j$ plays according to the strategy~$\strat_j$ and
  the other players~$i \not= j$ continue to play according to~$\sigma_i$.
    As the play~$\rho'$ is consistent with the \memoryless\ 
  winning strategy~$\strat_j$ from~$\vertu$, it visits~$\F_j$
  and we have
  \begin{align*}
    \payoff_j(\rho') & \leq  |h\vertu| + |V| &
    \text{(by Proposition~\ref{prop:gamei})} \\
    & \leq  (\nbband+1)\cdot|V| & \text{(by hypothesis)} \\
    & \leq  \depth & \text{(as~$\alpha\beta\gamma\leq \rho$).}
  \end{align*}
 
  We consider the following two cases.  If $\payoff_j(\rho)=+\infty$
  (i.e. $\rho$ does not visit~$\F_j$), we have
  $$\payoff_j(\rho') < \payoff_j(\rho) = +\infty.$$ On the contrary,
  if $\payoff_j(\rho)<+\infty$ (i.e. $\rho$ visits~$\F_j$, but after
  the prefix~$\alpha\beta\gamma$ by hypothesis), then we have
  $$\payoff_j(\rho') < \payoff_j(\rho)$$ as~$\payoff_j(\rho) >
  (\nbband+1)\cdot|V|$.
  
  Since~$\rho'$ is consistent with~$(\sigma_i)_{\ipimj}$, the 
  strategy of player~$j$ induced by the play~$\rho'$ is
  a profitable deviation for player~$j$ w.r.t.~$(\sigma_i)_{\ipi}$ in both 
  cases, which is a contradiction.
  \qed
\end{proof}

Now that we have proved Lemma~\ref{lemma:jeu djsn}, we use it in order
to obtain Lemma~\ref{lemma:NE same type}, which states that one can
define a Nash equilibrium~$(\tau_i)_{\ipi}$ in the game~$\ginf$, based
on the Nash equilibrium~$(\sigma_i)_{\ipi}$ in the game~$\gfin$.

\begin{lemma}\label{lemma:NE same type}
  Suppose~$d \geq 0$. Let~$(\sigma_i)_{\ipi}$ be a Nash equilibrium 
  in~$\gfin$ and~$\alpha\beta\gamma$ be a prefix of~$\rho=\langle
  (\sigma_i)_{\ipi} \rangle$ as defined in Lemma~\ref{lemma:jeu djsn}.
  Then there exists a Nash equilibrium~$(\tau_i)_{\ipi}$ in the
  game~$\ginf$.  Moreover~$(\tau_i)_{\ipi}$ is finite-memory, and
  $\valtype((\tau_i)_{\ipi})=\visit(\alpha)$.
\end{lemma}

\begin{proof}
  Let us set~$\Pi=\{1,\ldots,n\}$.  As~$\alpha$ and~$\beta$ end in the same
  vertex, we can consider the infinite play~$\abo$ in the
  game~$\ginf$.  Without loss of generality we can order the
  players~$i \in \Pi$ so that
  \begin{align*}
    \forall i \leq \jk & \ \ \ &  \payoff_i(\abo) < +\infty  
    &\ \ \ \  \text{($\alpha$ visits~$\F_i$)}\\
    \forall i > \jk & & \payoff_i(\abo) = +\infty  &\ \ \ \
    \text{($\alpha$ does not visit~$\F_i$)}
  \end{align*}
  where~$0 \leq \jk \leq n$.  In the second case, notice that~$\rho$
  could visit~$\F_i$ (but after the prefix~$\alpha\beta\gamma$).

  The Nash equilibrium~$(\tau_i)_{\ipi}$ required by
  Lemma~\ref{lemma:NE same type} is intuitively defined as follows.
  First the outcome of~$(\tau_i)_{\ipi}$ is exactly~$\abo$.  Secondly
  the first player~$j$ who deviates from~$\abo$ is punished by the
  coalition of the other players in the following way. If~$j \leq \jk$
  and the deviation occurs in the tree~$\tfin$, then the coalition
  plays according to~$(\sigma_i)_{\ipimj}$ in this tree. It prevents
  player~$j$ from reaching his goal set~$\F_j$ faster than in~$\abo$.
  And if~$j > \jk$, the coalition plays according
  to~$(\strat_{i,j})_{\ipimj}$ (given by Lemma~\ref{lemma:jeu djsn})
  so that player~$j$ does not reach his goal set at all.

  We begin by defining a punishment function~$\pun$ on the vertex
  set~$H$ of~$\tinf$ such that~$\pun(h)$ indicates the first
  player~$j$ who has deviated from~$\abo$, with respect to~$h$. We
  write~$\pun(h)= \bot$ if no deviation has occurred. For $v_0$, we
  define $\pun(v_0) = \bot$ and for~$h \in V^+$ such that~$\last(h)
  \in V_i$ and $v \in V$, we let:
  \[ \pun(hv) = \left\{
  \begin{array}{ll}
    \bot & \mbox{ if~$\pun(h)= \bot$ and~$hv < \abo$,}\\
    i & \mbox{ if~$\pun(h)= \bot$ and~$hv \not< \abo$,}\\
    \pun(h) & \mbox{ otherwise ($\pun(h) \not=\bot$)\,.} 
  \end{array}\right. \]
  
  The Nash equilibrium~$(\tau_i)_{\ipi}$ is then defined as follows:
  let~$h$ be a history ending in a vertex of~$V_i$,
  \begin{align}\label{eq:def taui}   \tau_i(h) = \left\{
  \begin{array}{ll}
    v & \mbox{ if~$\pun(h)=\bot$ ($h<\abo$); such
      that~$hv<\abo$,}\\ 
    \textit{arbitrary} & \mbox{
      if~$\pun(h)=i$,}\\ 
    \strat_{i,\pun(h)}(h) & \mbox{ if~$\pun(h)
      \not= \bot,i$ and~$\pun(h)>\jk$,}\\ 
    \sigma_i(h) & \mbox{
      if~$\pun(h) \not= \bot,i$, $\pun(h) \leq \jk$ and $|h| <
      \depth$,}\\ 
    \textit{arbitrary} & \mbox{ otherwise ($\pun(h)
      \not= \bot,i$, $\pun(h) \leq \jk$ and $|h| \geq \depth$)}
  \end{array}\right.\end{align}
  where \textit{arbitrary} means that the next vertex is chosen
  arbitrarily (in a \memoryless\ way). Clearly the outcome
  of~$(\tau_i)_{\ipi}$ is the play~$\abo$,
  and~$\valtype((\tau_i)_{\ipi})$ is equal to~$\visit(\alpha)$
  ($=\visit(\alpha\beta)$).
  
  It remains to prove that~$(\tau_i)_{\ipi}$ is a finite-memory Nash
  equilibrium in the game~$\ginf$.
  We first show that the strategy profile~$(\tau_i)_{\ipi}$ defined in
  Equation~\eqref{eq:def taui} is a Nash equilibrium in the
  game~$\ginf$.
  Let~$\tau_j'$ be a strategy of player~$j$. We show that this is not
  a profitable deviation for player~$j$ w.r.t.~$(\tau_i)_{\ipi}$. We
  distinguish the following two cases:
  \begin{enumerate}
  \item[$(i)$] $j \leq \jk$ ($\payoff_j(\abo) < +\infty$,~$\alpha$
    visits~$\F_j$).
    
    \smallskip To improve his cost, player~$j$ has no incentive to
    deviate after the prefix~$\alpha$. Thus we assume that the
    strategy~$\tau_j'$ causes a deviation from a vertex visited
    in~$\alpha$. By Equation~\eqref{eq:def taui} the other players
    first play according to~$(\sigma_i)_{\ipimj}$ in~$\gfin$, and then
    in an arbitrary way.

    Suppose that~$\tau_j'$ is a profitable deviation for player~$j$
    w.r.t.~$(\tau_i)_{\ipi}$ in the game~$\ginf$.  Let us set~$\playtau=
    \langle (\tau_i)_{\ipi} \rangle$ and~$\playdev= \langle
    \tau_j',(\tau_i)_{\ipimj} \rangle$. Then
    $$\payoff_j(\playdev) < \payoff_j(\playtau).$$ On the other hand
    we know that
    $$\payoff_j(\playtau) = \payoff_j(\rho) \leq |\alpha|.$$ So if we
    limit the play~$\playdev$ in~$\ginf$ to its prefix of
    length~$\depth$, we get a play~$\rho'$ in~$\gfin$ such that
    $$\payoff_j(\rho')=\payoff_j(\playdev) < \payoff_j(\rho).$$ As the
    play~$\rho'$ is consistent with the
    strategies~$(\sigma_i)_{\ipimj}$ by Equation~\eqref{eq:def taui},
    the strategy~$\tau_j'$ restricted to the tree~$\tfin$ is a
    profitable deviation for player~$j$ w.r.t.~$(\sigma_i)_{\ipi}$ in
    the game~$\gfin$. This contradicts the fact
    that~$(\sigma_i)_{\ipi}$ is a Nash equilibrium in this game.
    \smallskip
  \item[$(ii)$] $j > \jk$ ($\payoff_j(\abo) = +\infty$,~$\abo$ does not
    visit~$\F_j$).

    \smallskip If player~$j$ deviates from~$\abo$ (with the
    strategy~$\tau_j'$), by Equation~\eqref{eq:def taui} the other
    players combine against him and play according to~$\strat_{-j}$.
    By Lemma~\ref{lemma:jeu djsn} this coalition wins the
    game~$\mathcal{G}_j$ from any vertex visited by~$\abo$.  So the
    strategy~$\strat_{-j}$ of the coalition keeps the play~$\langle
    \tau_j', (\tau_i)_{\ipimj}\rangle$ away from the set~$\F_j$,
    whatever player~$j$ does. Therefore~$\tau_j'$ is not a profitable
    deviation for player~$j$ w.r.t.~$(\tau_i)_{\ipi}$ in the
    game~$\ginf$.
  \end{enumerate}

  We now prove that $(\tau_i)_{\ipi}$ is a finite-memory strategy profile.
According to the definition of finite-memory strategy (see
Section~\ref{sec:prelim}) we have to prove that each 
relation~$\approx_{\tau_i}$ on~$H$ has finite index (recall that
$h \approx_{\tau_i} h'$ if $\tau_i(h\delta)=\tau_i(h'\delta)$ for 
all~$\delta \in V^*V_i$).
In this aim we define for each 
player~$i$ an equivalence relation~$\sim_{\tau_i}$ with finite index such
that $$\forall h,h' \in H,\ \ \ 
h \sim_{\tau_i} h' \Rightarrow h \approx_{\tau_i} h'.$$

We first define an equivalence relation~$\sim_{\pun}$ with finite index related
to the punishment function~$\pun$. For all prefixes~$h$,~$h'$ 
of~$\alpha\beta^{\omega}$, i.e. such that no player is punished, 
this relation does not distinguish two histories that are
identical except for a certain number of cycles~$\beta$. 
For the other histories it just has to remember the first player, say~$i$,
who has deviated. The definition of~$\sim_{\pun}$ is as follows:
\begin{align*}
  & h \sim_{\pun} h' & &
  \text{ if } h = \alpha\beta^l\beta',\ h'= \alpha\beta^m\beta',\ 
   \beta' < \beta,\ l,m \geq 0\\
  & hv \sim_{\pun} h'v' & & \text{ if } v,v' \in V_i ,\ 
  h,h' < \abo, \text{ but } hv, h'v' \not < \abo\\
  & hv \sim_{\pun} hv\delta & &\text{ if } h < \abo,\ hv \not< \abo,\ 
  \delta \in V^*.
\end{align*}
The relation~$\sim_{\pun}$ is an equivalence relation on~$H$ with finite
index.

We now turn to the definition of~$\sim_{\tau_i}$. It is based on the
definition of~$\tau_i$ (given in~\eqref{eq:def taui})
and~$\sim_{\pun}$. To get an equivalence with finite index we proceed
as follows. Recall that each strategy~$\strat_{i,\pun(h)}$ is \memoryless\
and when a player plays arbitrarily, his strategy is also \memoryless.
Furthermore notice that, in the definition of~$\tau_i$, 
the strategy~$\sigma_i$ is 
only applied to histories~$h$ with length~$|h| < \depth$. 
For histories~$h$ such that~$\tau_i(h)=v$ with~$hv < \abo$, 
it is enough to remember information
with respect to~$\alpha\beta$ as already done for~$\sim_{\pun}$. 
Therefore for~$h,h' \in H$ we define~$\sim_{\tau_i}$ 
in the following way:
\begin{align*}
  h \sim_{\tau_i} h' & \text{\ \ if\ \ } h \sim_{\pun} h' & \text {and\ \ }
  &  \big(\pun(h)=\bot \\
  & & & \text{ or } \pun(h) = i \text{ and } \last(h)=\last(h')\\
  & & & \text{ or } \pun(h) \not= \bot, i,\ \pun(h) > \jk  \text{ and } 
  \last(h)=\last(h')\\
  & & & \text{ or } \pun(h) \not= \bot, i,\ \pun(h) \leq \jk,\ 
  |h|,|h'| \ge \depth \text{ and}\\
  & & & \ \ \ \ \  \last(h)=\last(h')\big).
\end{align*}
Notice that this relation satisfies
$$h \sim_{\tau_i} h' \Rightarrow \tau_i(h)=\tau_i(h') \text{ and }
\last(h)=\last(h')$$
and has finite index.
Therefore if~$h \sim_{\tau_i} h'$, then~$h \approx_{\tau_i} h'$ and the 
relation~$\approx_{\tau_i}$ has finite index.

  \qed
\end{proof}

We can now proceed to the proof of Proposition~\ref{prop:NE fin ds inf}.

\begin{proof}[of Proposition~\ref{prop:NE fin ds inf}]
  Let us set~$\Pi=\{1,\ldots,n\}$ and $\depth=(n+1)\cdot 2\cdot|V|$.
  Let~$(\sigma_i)_{\ipi}$ be a Nash equilibrium in the game~$\gfin$
  and~$\rho$ its outcome.

  To be able to use Lemma~\ref{lemma:NE same
    type}, we consider the prefix~$\la\lap$ of~$\rho$ of minimal
  length such that
  \begin{align}
    \exists \, \nbband \geq 1 \ \ \ \ \ \ & |\la| = (\nbband-1)\cdot|V|\notag\\
    & |\la\lap|=(\nbband+1)\cdot|V|\notag\\
    & \visit(\la)=\visit(\la\lap)\,. \label{eq:v(l)=v(ll')}
  \end{align}
   The following statements are true.
  \begin{itemize}
  \item[$\bullet$] $\nbband \leq 2\cdot n+1$.
  \item[$\bullet$] If~$\visit(\la) \varsubsetneq \visit(\rho)$,
    then~$\nbband < 2\cdot n+1$.
  \end{itemize}
  Indeed the first statement results from the fact that in the worst
  case, the play~$\rho$ visits the goal set of a new player in each
  prefix of length~$i\cdot2\cdot|V|$,~$1 \leq i \leq n$, i.e.~$|\la|=
  n\cdot 2\cdot|V|$. It follows that~$\la\lap$ exists as a prefix
  of~$\rho$, because the length~$\depth$ of~$\rho$ is equal
  to~$(n+1)\cdot 2\cdot|V|$ by hypothesis.  Thus~$\visit(\la)\subseteq
  \visit(\rho)$.  Suppose that there exists~$i \in \visit(\rho)
  \setminus \visit(\la)$, then~$\rho$ visit~$\F_i$ after the
  prefix~$\la\lap$ by Equation~\eqref{eq:v(l)=v(ll')}. The second
  statement follows easily.

  Given the length of~$\lap$, one vertex of~$V$ is visited at least
  twice by~$\lap$. More precisely, we can write
  \begin{align*}
    \la\lap = \alpha\beta\gamma  \text{\ \ \  with\ \ \ }
    & \last(\alpha)= \last(\alpha\beta)\notag\\
    & |\alpha| \geq (\nbband -1) \cdot |V| \notag\\
    & |\alpha\beta| \leq \nbband\cdot|V|\,.
  \end{align*}
  In particular,~$|\la| \leq |\alpha|$. See Figure~\ref{fig:slicing
    play}.  We have~$\visit(\alpha)=\visit(\alpha\beta\gamma)$, and
  $|\alpha\beta\gamma| =(\nbband+1)\cdot|V|$.




\begin{figure}[h!]
    \centering
    \begin{tikzpicture}[yscale=.5,xscale=.6]
      \everymath{\scriptstyle}
      \draw (0,8) -- (-5,0);
      \draw (0,8) -- (5,0);

      \draw[fill=black] (0,8) circle (2pt);

      \draw (-5,0) -- (5,0);

      \draw[very thin,dotted] (-4,2) -- (5.1,2);
      \draw[very thin,dotted] (-4,3) -- (5.1,3);
      \draw[very thin,dotted] (-4,4) -- (5.1,4);
      \draw[very thin,dotted] (-4,6) -- (5.1,6);
      \draw[very thin,dotted] (-4,7) -- (5.1,7);
      \draw[very thin,dotted] (-4,8) -- (5.1,8);

      \path (5.1,7) node[right] (q0) {$|V|$};
      \path (5.1,6) node[right] (q0) {$2\cdot|V|$};
      \path (5.1,4) node[right] (q0) {$(\nbband-1)\cdot|V|$};
      \path (5.1,3) node[right] (q0) {$\nbband\cdot|V|$};
      \path (5.1,2) node[right] (q0) {$(\nbband+1)\cdot|V|$};
      \path (5.2,0) node[right] (q0) {$\depth$};

      \draw[thick, dashed] (0,8) .. controls (-1,6) .. (.15,3.7);
      \draw[fill=black] (.15,3.7) circle (2pt);
      \path (-1.2,5) node[right] (q0) {$\alpha$};

      \draw[thick] (.1,3.7) .. controls (0.2,3.6) .. (.4,3.2);
      \draw[fill=black] (.4,3.2) circle (2pt);
      \path (.4,3.5) node[right] (q0) {$\beta$};

      \draw[thick, densely dotted] (.4,3.2) .. controls (.8,2.5) .. (1.2,2);
      \draw[fill=black] (1.2,2) circle (2pt);
      \path (.6,2.4) node[left] (q0) {$\gamma$};

      \draw (1.2,2) .. controls (1.5,1.5) .. (2.1,0);

      \draw (2.1,-.1) node[below] (q0) {$\rho$};

      \draw [snake=brace] (0,7.9) -- (1.2,4) node [right,pos=.45]
            {${} \, \la$};

      \draw [snake=brace] (1.3,4) -- (2.1,2) node [right,pos=.44]
            {${} \, \lap$};

    \end{tikzpicture}
    \caption{Slicing of the play $\rho$ in the tree $\tfin$.}
    \label{fig:slicing play}
  \end{figure}


  As the hypotheses of Lemma~\ref{lemma:NE same type} are verified, we
  can apply it in this context to get a finite-memory Nash
  equilibrium~$(\tau_i)_{\ipi}$ in the game~$\ginf$
  with~$\valtype((\tau_i)_{\ipi})=\visit(\alpha)$.  \qed
\end{proof}




Proposition~\ref{prop:NE fin ds inf} asserts that given a
game~$\mathcal{G}$ and the game~$\gfin$ played on the truncated tree
of~$\tinf$ of a well-chosen depth~$\depth$, one can lift any Nash
equilibrium~$(\sigma_i)_{\ipi}$ of~$\gfin$ to a Nash
equilibrium~$(\tau_i)_{\ipi}$ of~$\mathcal{G}$. The proof of
Proposition~\ref{prop:NE fin ds inf} states that the \type\
of~$(\tau_i)_{\ipi}$ is equal to~$\visit(\alpha)$. We give an example
that shows that it is impossible to preserve the \type\ of the lifted
Nash equilibrium~$(\sigma_i)_{\ipi}$.  
  \begin{example}\label{ex:diff-types}
  Let us consider the two-player game~$\mathcal{G}$ depicted in
  Figure~\ref{fig-extypedifun} with $\F_1=\{C\}$,~$\F_2=\{E\}$. 
  One can show that~$\mathcal{G}$ admits
  only Nash equilibria of \type~$\{2\}$ or~$\emptyset$. 
  Indeed, on one hand, there is no play of~$\mathcal{G}$
  where both goals are visited, and on the other hand given
  a strategy profile~$(\sigma_i)_{\ipi}$ such that~$ \langle
  (\sigma_i)_{\ipi} \rangle$ visits~$\F_1$ (i.e.~$ \langle
  (\sigma_i)_{\ipi} \rangle$ is of the form~$A^+BC^\omega$), playing~$D$ 
  instead of~$C$ is clearly a profitable deviation for
  player~$2$. 

  We will now see that for each~$\depth \geq 2$ the game played
  on~$\tfin$ admits a Nash equilibrium of \type~$\{1\}$. From the
  above discussion, this equilibrium can not be lifted to a Nash
  equilibrium of the same \type\ in~$\mathcal{G}$. A truncated
  tree~$\tfin$ is depicted in Figure~\ref{fig-extypedifdeux}. One can
  show that the strategy profile leading to the
  outcome~$A^{\depth-1}BC$ (depicted in bold in the figure) is a Nash
  equilibrium in~$\gfin$ of \type~$\{1\}$. Following the lines of the
  proof of Proposition~\ref{prop:NE fin ds inf}, we see that this Nash
  equilibrium is lifted to a Nash equilibrium of~$\mathcal{G}$ with
  outcome~$A^{\omega}$ and \type~$\emptyset$.

\begin{figure}[h!]
  \null\hfill
  \begin{minipage}[b]{0.35\linewidth}
    \centering
    \begin{tikzpicture}[xscale=.7]
      \everymath{\scriptstyle}

      \path (0,0) node[draw,circle,inner sep=2pt] (A) {$A$};
      \path (2,0) node[draw,rectangle,inner sep=4pt] (B) {$B$};
      \path (4,0) node[draw,circle,inner sep=2pt,fill=black!20!white] (C) 
{$C$};
      \path (2,-2) node[draw,rectangle,inner sep=4pt] (D) {$D$};
      \path (4,-2) node[draw,circle,inner sep=2pt,double] (E) {$E$};

\phantom{
      \path (4,-3.5) node[draw,circle,inner sep=2pt] (phantom) {$E$};
}


      \draw[arrows=-latex'] (A) .. controls +(135:.75cm) and
      +(45:.75cm) .. (A);

      \draw[arrows=-latex'] (C) .. controls +(135:.75cm) and
      +(45:.75cm) .. (C);

      \draw[arrows=-latex'] (E) .. controls +(135:.75cm) and
      +(45:.75cm) .. (E);

      \draw[arrows=-latex'] (-.75,0) -- (A);
      \draw[arrows=-latex'] (A) -- (B);
      \draw[arrows=-latex'] (B) -- (C);
      \draw[arrows=-latex'] (B) -- (D);
      \draw[arrows=-latex'] (D) -- (E);

    \end{tikzpicture}
    \caption{A game~$\mathcal G$.}
    \label{fig-extypedifun}
  \end{minipage}
  \hfill
  \begin{minipage}[b]{0.62\linewidth}
    \centering
    \begin{tikzpicture}[xscale=.27,yscale=.6]
      \everymath{\scriptscriptstyle}

      \draw[very thin] (0,10.5) -- (-7,8.5) -- (-12.5,1.5) --
      (12.5,1.5) -- (7,8.5) -- (0,10.5);

      \path (12.5,1.5) node[right] (dd) {$\depth$};

\begin{scope}
      \path (0,10) node[] (11) {$\mathbf{A}$};

      \path (-4,9) node[] (21) {$\mathbf{A}$};
      \path (4,9) node[] (22) {$B$};

      \path (-6,8) node[] (31) {$\mathbf{A}$};
      \path (-2,8) node[] (32) {$B$};
      \path (2,8) node[] (33) {$C$};
      \path (6,8) node[] (34) {$D$};

      \path (-7,7) node[] (41) {$\mathbf{A}$};
      \path (-5,7) node[] (42) {$B$};
      \path (-3,7) node[] (43) {$C$};
      \path (-1,7) node[] (44) {$D$};
      \path (2,7) node[] (45) {$C$};
      \path (6,7) node[] (46) {$E$};

      \draw[thick] (11)--(21);
      \draw[very thin] (11)--(22);

      \draw[thick] (21)--(31);
      \draw[very thin] (21)--(32);
      \draw[very thin] (22)--(33);
      \draw[very thin] (22)--(34);

      \draw[thick] (31)--(41);
      \draw[very thin] (31)--(42);
      \draw[very thin] (32)--(43);
      \draw[very thin] (32)--(44);
      \draw[very thin] (33)--(45);
      \draw[very thin] (34)--(46);

      \draw[thick,densely dotted] (41) -- (-8,5.5);
      \draw[densely dotted] (41) -- (-6.2,6);
      \draw[densely dotted] (42) -- (-5.8,6);
      \draw[densely dotted] (42) -- (-4.2,6);
      \draw[densely dotted] (43) -- (-3,6);
      \draw[densely dotted] (44) -- (-1,6);
      \draw[densely dotted] (45) -- (2,6);
      \draw[densely dotted] (46) -- (6,6);
\end{scope}

\begin{scope}[xshift=-8.25cm,yshift=-5cm,xscale=.4]
      \path (0,10) node[] (11) {$\mathbf{A}$};

      \path (-4,9) node[] (21) {$\mathbf{A}$};
      \path (4,9) node[] (22) {$B$};

      \path (-6,8) node[] (31) {$A$};
      \path (-2,8) node[] (32) {$\mathbf{B}$};
      \path (2,8) node[] (33) {$C$};
      \path (6,8) node[] (34) {$D$};

      \path (-7.5,7) node[] (41) {$A$};
      \path (-5.5,7) node[] (42) {$B$};
      \path (-2.8,7) node[] (43) {$\mathbf{C}$};
      \path (-.8,7) node[] (44) {$D$};
      \path (2,7) node[] (45) {$C$};
      \path (6,7) node[] (46) {$E$};

      \draw[thick] (11)--(21);
      \draw[very thin] (11)--(22);

      \draw[very thin] (21)--(31);
      \draw[thick] (21)--(32);
      \draw[very thin] (22)--(33);
      \draw[very thin] (22)--(34);

      \draw[very thin] (31)--(41);
      \draw[very thin] (31)--(42);
      \draw[thick] (32)--(43);
      \draw[very thin] (32)--(44);
      \draw[very thin] (33)--(45);
      \draw[very thin] (34)--(46);
\end{scope}

    \end{tikzpicture}
    \caption{The truncated tree~$\tfin$.}
    \label{fig-extypedifdeux}
  \end{minipage}
  \hfill\null
\end{figure}

\end{example}

On the other hand, notice that from the proof of
Proposition~\ref{prop:NE fin ds inf}, we can construct a Nash
equilibrium such that each player pays either an infinite cost, or a
cost bounded by~$|\Pi|\cdot 2 \cdot |V|$.

\subsection{Nash Equilibria with Finite Memory Preserving
  Types}\label{sub:fin-mem}

In this section we show that given a Nash equilibrium, we can
construct another Nash equilibrium with the same \type\ such that all
its strategies are finite-memory. We then answer to Problem~\ref{pbm
  2} for Nash equilibria.

\begin{theorem}\label{theo:fin-mem}
  If there exists a Nash equilibrium in a quantitative multiplayer
  reachability game~$\mathcal{G}$, then there exists a finite-memory
  Nash equilibrium of the same \type\ in~$\mathcal{G}$.
\end{theorem}

The proof is based on two steps. The first step constructs from
$(\sigma_i)_{\ipi}$ ano\-ther Nash equilibrium~$(\tau_i)_{\ipi}$ with the
same \type\ such that the play~$\langle (\tau_i)_{\ipi} \rangle$ is
of the form~$\abo$ with~$\visit(\alpha)=\valtype((\sigma_i)_{\ipi})$.
This is possible thanks to
Lemmas~\ref{lemma 1} and~\ref{lemma 2},
by first eliminating unne\-ces\-sary cycles in the play~$\langle
(\sigma_i)_{\ipi} \rangle$ and then locating a prefix~$\alpha\beta$
such that~$\beta$ is a cycle that can be infinitely repeated.

The second step transforms the Nash equilibrium~$(\tau_i)_{\ipi}$ into
a finite-memory one thanks to Lemma~\ref{lemma:NE same type} given in
Section~\ref{sub:ex NE}.  For that purpose, we consider the strategy
profile~$(\tau_i)_{\ipi}$ limited to the tree~$\tinf$ truncated at a
well-chosen depth.



The next lemma indicates how to eliminate a cycle in the outcome of a
Nash equilibrium.

\begin{lemma}\label{lemma 1}
  Let~$(\sigma_i)_{\ipi}$ be a strategy profile in a game~$\mathcal{G}$
  and~$\rho = \langle 
  (\sigma_i)_{\ipi} \rangle$ its outcome.
  Suppose that~$\rho = \la\lap\rhobar$, where~$\lap$ contains at least
  one vertex, such that
  \begin{align*}
    & \visit(\la)=\visit(\la\lap)\\
    & \last(\la)=\last(\la\lap).
  \end{align*}
  We define a strategy profile~$(\tau_i)_{\ipi}$  as follows: \
  \[ \tau_i(h) = \left\{
  \begin{array}{ll}
    \sigma_i(h) &
  \mbox{ if~$\la \not\leq h$,}\\
  \sigma_i(\la\lap\delta) & \mbox{ if~$h=\la\delta$} 
  \end{array}\right. \]
  where~$h$ is a history of~$\mathcal{G}$ with~$\last(h) \in V_i$.
  We get the outcome~$\langle (\tau_i)_{\ipi} \rangle = \la\rhobar$.\\
  \noindent If a strategy~$\tau_j'$ is a profitable deviation for player~$j$
  w.r.t.~$(\tau_i)_{\ipi}$,
  then there exists a profitable deviation~$\sigma_j'$ for player~$j$
  w.r.t.~$(\sigma_i)_{\ipi}$.
\end{lemma}

\begin{proof}
Let us set~$\Pi=\{1,\ldots,n\}$. We write
\begin{align*}
  \rho = \langle (\sigma_i)_{\ipi} \rangle & \text{ of cost profile }
  (x_1,\ldots,x_n),\\
  \pi = \langle (\tau_i)_{\ipi} \rangle & \text{ of cost profile }
  (y_1,\ldots,y_n).
\end{align*}

We observe that as~$\rho = \la\lap\rhobar$, we have~$\pi=\la\rhobar$
(see Figures~\ref{fig:loop} and~\ref{fig:cut loop}). It follows that 
\begin{align}
  \forall \ipi,\ \ y_i \leq x_i. \label{eq:y_i + pt x_i}
\end{align} 
More precisely,
\begin{align}
  & \text{- if } x_i = +\infty, \text{ then } y_i=+\infty; \label{eq:x_i=y_i}\\
  & \text{- if } x_i < +\infty \text{ and } i \in \visit(\la), 
  \text{ then } y_i=x_i; \notag\\
  & \text{- if } x_i < +\infty \text{ and } i \not \in \visit(\la), 
    \text{ then } y_i=x_i-(|\lap|+1). \label{eq:x_i>y_i}
\end{align}

\begin{figure}[h!]
  \null\hfill
  \begin{minipage}[b]{0.48\linewidth}
    \centering
    \begin{tikzpicture}[yscale=.5,xscale=.5]
      \everymath{\scriptstyle}
      \draw (0,10) -- (-5,0);
      \draw (0,10) -- (5,0);

      \draw[black!20!white,fill=black!20!white] (.5,7) -- (4,0) -- (-3.5,0);
      \draw[black!10!white,fill=black!10!white] (-.5,4) -- (1.5,0) -- (-2.5,0);

      \draw[fill=black] (0,10) circle (2pt);
      \draw[thick] (0,10) .. controls (-.5,8.5) .. (.5,7);

      \draw[fill=black] (.5,7) circle (2pt);
      \draw[] (.5,7) .. controls (1.2,5) and (-.2,5) .. (-.5,4);

      \draw[fill=black] (-.5,4) circle (2pt);
      \draw[thick] (-.5,4) .. controls (-1,3) and (-.5,2) .. (-1.5,0);

      \draw[fill=black] (-.2,8.1) circle (2pt);
      \draw[thick,dashed] (-.2,8.1) .. controls (-3,5) and (-1,2) .. (-4.5,0);

      \draw[fill=black] (-.825,2.3) circle (2pt);
      \draw[thick,dashed] (-.825,2.3) .. controls (1,.8)  .. (1.2,0);

      \draw (-1.5,-.1) node[below] (q0) {$\rho = \langle
        (\sigma_i)_{\ipi} \rangle$};
      \draw (-4.5,.1) node[below] (q0) {$\rho_1'$};
      \draw (1.3,.1) node[below] (q0) {$\rho_2'$};

      \draw (0,9) node[below] (q0) {$\la$};
      \draw (1,6) node[below] (q0) {$\lap$};
      \draw (-1.3,2) node[below] (q0) {$\rhobar$};
    \end{tikzpicture}
    \caption{Play $\rho$ and possible deviations.}
    \label{fig:loop}
  \end{minipage}
  \hfill
  \begin{minipage}[b]{0.48\linewidth}
    \centering
    \begin{tikzpicture}[yscale=.5,xscale=.5]
      \everymath{\scriptstyle}
      \draw (0,10) -- (-5,0);
      \draw (0,10) -- (5,0);

      \draw[fill=black] (0,10) circle (2pt);
      \draw[thick] (0,10) .. controls (-.5,8.5) .. (.5,7);

      \draw[fill=black] (-.2,8.1) circle (2pt);
      \draw[thick,dashed] (-.2,8.1) .. controls (-3,5) and (-1,2) .. (-4.5,0);
 
      \draw (-.5,1.5) node[below] (q0) {$\pi = \langle
        (\tau_i)_{\ipi} \rangle$};
      \draw (-4.5,.1) node[below] (q0) {$\pi_1'$};
      \draw (2.2,1.5) node[below] (q0) {$\pi_2'$};

      \draw (0,9) node[below] (q0) {$\la$};
      \draw (1,6) node[below] (q0) {$\lap$};

\begin{scope}[xshift=1cm,yshift=3cm]
      \draw[white,fill=black!10!white] (-.5,4) -- (2,-1) -- (-3,-1);

      \draw[fill=black] (-.5,4) circle (2pt);
      \draw[thick] (-.5,4) .. controls (-1,3) and (-.5,2) .. (-1.5,0);
      \draw[thick,dotted] (-1.5,0) .. controls  (-1.7,-.75) ..  (-1.7,-1.5);

      \draw[fill=black] (-.825,2.3) circle (2pt);
      \draw[thick,dashed] (-.825,2.3) .. controls (1,.8)  .. (1.2,0);
      \draw[thick,dotted] (1.2,0) .. controls  (1.2,-.75) ..  (1.1,-1.5);
      \draw (-1.3,2) node[below] (q0) {$\rhobar$};
\end{scope}
    \end{tikzpicture}
    \caption{Play $\pi$ and possible deviations.}
    \label{fig:cut loop}
  \end{minipage}
  \hfill\null
\end{figure}

Let~$\tau_j'$ be a profitable deviation for player~$j$ 
w.r.t.~$(\tau_i)_{\ipi}$, and~$\pi'$ be the outcome of the 
strategy profile~$(\tau_j', (\tau_i)_{\ipimj})$.
Then $$\payoff_j(\pi') < y_j.$$
We show how to construct a profitable deviation~$\sigma_j'$ for player~$j$
w.r.t.~$(\sigma_i)_{\ipi}$. Two cases occur:
\begin{enumerate}
\item[$(i)$] player~$j$ deviates from~$\pi$ just after a proper prefix~$h$ 
  of~$\la$ (like for the play~$\pi'_1$ in Figure~\ref{fig:cut loop}). 

  \smallskip
  We define~$\sigma_j'=\tau_j'$ and we denote by~$\rho'$ the outcome 
  of~$(\sigma_j',(\sigma_i)_{\ipimj})$. Given the definition of
  the strategy profile~$(\tau_i)_{\ipi}$, one can verify that~$\rho'=\pi'$
  (see the play~$\rho_1'$ in Figure~\ref{fig:loop}).
  Thus
  $$\payoff_j(\rho')=\payoff_j(\pi') < y_j \leq x_j$$
  by Equation~\eqref{eq:y_i + pt x_i}, which implies that~$\sigma_j'$
  is a profitable deviation of player~$j$ w.r.t. $(\sigma_i)_{\ipi}$.

\smallskip

\item[$(ii)$] player~$j$ deviates from~$\pi$ after the prefix~$\la$
  ($\pi$ and~$\pi'$ 
  coincide at least on~$\la$). 

  \smallskip
  This case is illustrated by the play~$\pi_2'$ 
  in Figure~\ref{fig:cut loop}. We define for all histories~$h$
  ending in a vertex of~$V_j$:
  \[ \sigma_j'(h) = \left\{
  \begin{array}{ll}
  \sigma_j(h) &  \mbox{ if~$\la\lap \not\leq h$,}\\
  \tau_j'(\la\delta) & \mbox{ if~$h=\la\lap\delta$.} 
\end{array}\right. \]
Let us set~$\rho'=\langle \sigma_j',(\sigma_i)_{\ipimj} \rangle$. As
player~$j$ deviates after~$\la$ with the strategy~$\tau_j'$, one can
prove that
$$\pi'=\la\pibar' \quad \text{and} \quad \rho'=\la\lap\pibar'$$
by definition of~$(\tau_i)_{\ipi}$ (see the play~$\rho_2'$ in 
Figure~\ref{fig:loop}). As~$\payoff_j(\pi') < y_j$, it means
that~$j \not \in \visit(\la)$ (other\-wise the deviation would not be
profitable for player~$j$). Since~$\visit(\la)=\visit(\la\lap)$,
we also have
$$\payoff_j(\pi')+(|\lap|+1) = \payoff_j(\rho').$$
By Equations~\eqref{eq:x_i=y_i} and~\eqref{eq:x_i>y_i}, we get
\begin{itemize}
  \item either~$x_j=y_j=+\infty$ and~$\payoff_j(\rho') < x_j$,
  \item or~$x_j=y_j + (|\lap|+1)$ and~$\payoff_j(\rho') < x_j$,
\end{itemize}
which proves that~$\sigma_j'$ is a profitable deviation for player~$j$
w.r.t.~$(\sigma_i)_{\ipi}$.\qed
\end{enumerate}

\end{proof}

While Lemma~\ref{lemma 1} deals with elimination of unnecessary cycles,
Lemma~\ref{lemma 2} deals with repetition of a useful cycle.

\begin{lemma}\label{lemma 2}
  Let~$(\sigma_i)_{\ipi}$ be a strategy profile in a game~$\mathcal{G}$
  and~$\rho = \langle (\sigma_i)_{\ipi} \rangle$ its outcome.
  We assume that~$\rho = \la\lap\rhobar$, where~$\lap$ contains at least
  one vertex, such that
  \begin{align*}
    & \visit(\la)=\visit(\rho)\\
    & \last(\la)=\last(\la\lap).
  \end{align*}
  We define a strategy profile~$(\tau_i)_{\ipi}$  as follows: \
  \[ \tau_i(h) = \left\{
  \begin{array}{ll}
    \sigma_i(h) & \mbox{ if~$\la \not\leq h$,}\\
    \sigma_i(\la\delta) & \mbox{ if~$h=\la\lap^k\delta$,~$k \in \N$, and
     ~$\lap \not \leq \delta$} 
  \end{array}\right. \]
  where~$h$ is a history of~$\mathcal{G}$ with~$\last(h) \in V_i$.
  We get the outcome~$\langle (\tau_i)_{\ipi} \rangle = \la\lap^{\omega}$.\\
  \noindent  If a strategy~$\tau_j'$ is a profitable deviation for player~$j$
  w.r.t.~$(\tau_i)_{\ipi}$,
  then there exists a profitable deviation~$\sigma_j'$ for player~$j$
  w.r.t.~$(\sigma_i)_{\ipi}$.
\end{lemma}
  
\begin{proof}
  We use the same notations as in the proof of Lemma~\ref{lemma 1}.
  Here we have~$x_i=y_i$ for all~$\ipi$ since~$\visit(\la)=\visit(\rho)$.
  One can prove that~$\pi=\la\lap^{\omega}$ 
  (see Figures~\ref{fig:1 loop} and~\ref{fig:inf loop}).

  \begin{figure}[h!]
  \null\hfill
  \begin{minipage}[b]{0.48\linewidth}
    \centering
    \begin{tikzpicture}[yscale=.5,xscale=.5]
      \everymath{\scriptstyle}
      \draw (0,10) -- (-5,0);
      \draw (0,10) -- (5,0);

      \draw[black!20!white,fill=black!20!white] (0,8) -- (4,0) -- (-4,0);
      \draw[white,fill=white] (0,6) -- (3,0) -- (-3,0);

      \draw[fill=black] (0,10) circle (2pt);
      \draw[] (0,10) .. controls (-.3,9) .. (0,8);

      \draw[fill=black] (0,8) circle (2pt);
      \draw[thick] (0,8) .. controls (.4,7) and (-.2,6.5) .. (0,6);

      \draw[fill=black] (0,6) circle (2pt);
      \draw[dashed,thick] (0,6) .. controls (1,3) and (-.5,2) .. (-1.5,0);

      \draw (-.3,9) node[right] (q0) {$\la$};
      \draw (-.05,6.8) node[right] (q0) {$\lap$};
      \draw (.5,3) node[below] (q0) {$\rhobar$};
      \draw (-1.5,-.1) node[below] (q0) {$\rho = \langle
        (\sigma_i)_{\ipi} \rangle$};
    \end{tikzpicture}
    \caption{Play $\rho$ and its prefix $\la\lap$.}
    \label{fig:1 loop}
  \end{minipage}
  \hfill
  \begin{minipage}[b]{0.48\linewidth}
    \centering
    \begin{tikzpicture}[yscale=.5,xscale=.5]
      \everymath{\scriptstyle}
      \draw (0,10) -- (-5,0);
      \draw (0,10) -- (5,0);

      \draw[black!20!white,fill=black!20!white] (0,8) -- (4,0) -- (-4,0);
      \draw[black!10!white,fill=black!10!white] (0,6) -- (3,0) -- (-3,0);
      \draw[black!20!white,fill=black!20!white] (0,4) -- (2,0) -- (-2,0);
      \draw[black!10!white,fill=black!10!white] (0,2) -- (1,0) -- (-1,0);

      \draw[fill=black] (0,10) circle (2pt);
      \draw[] (0,10) .. controls (-.3,9) .. (0,8);

      \draw[fill=black] (0,8) circle (2pt);
      \draw[thick] (0,8) .. controls (.4,7) and (-.2,6.5) .. (0,6);

      \draw (-.3,9) node[right] (q0) {$\la$};
      \draw (-.05,6.8) node[right] (q0) {$\lap$};

      \draw (0,-.1) node[below] (q0) {$\pi = \langle
        (\tau_i)_{\ipi} \rangle$};

      \begin{scope}[yshift=-2cm]
      \draw[fill=black] (0,8) circle (2pt);
      \draw[thick] (0,8) .. controls (.4,7) and (-.2,6.5) .. (0,6);
      \draw (-.05,6.8) node[right] (q0) {$\lap$};
      \end{scope}

      \begin{scope}[yshift=-4cm]
      \draw[fill=black] (0,8) circle (2pt);
      \draw[thick] (0,8) .. controls (.4,7) and (-.2,6.5) .. (0,6);
      \draw (-.05,6.8) node[right] (q0) {$\lap$};
      \end{scope}

      \begin{scope}[yshift=-6cm]
      \draw[fill=black] (0,8) circle (2pt);
      \draw[thick] (0,8) .. controls (.4,7) and (-.2,6.5) .. (0,6);
      \draw (-.05,6.8) node[right] (q0) {$\lap$};
      \end{scope}

    \end{tikzpicture}
    \caption{Play $\pi=\la\lap^{\omega}$.}
    \label{fig:inf loop}
  \end{minipage}
  \hfill\null
\end{figure}

  We show how to define a profitable deviation~$\sigma_j'$ from
  the deviation~$\tau_j'$. We distinguish the following two cases:
  \begin{enumerate}
  \item[$(i)$] player~$j$ deviates from~$\pi$ just after a proper prefix~$h$
    of~$\la\lap$.

    \smallskip
    We define~$\sigma_j'=\tau_j'$. As in the first case of the proof
    of Lemma~\ref{lemma 1}, we have~$\payoff_j(\rho') < x_j$, 
    which implies that~$\sigma_j'$ is
    a profitable deviation of player~$j$ w.r.t.~$(\sigma_i)_{\ipi}$.
  \item[$(ii)$] player~$j$ deviates from~$\pi$ after the prefix~$\la\lap$, i.e.
    after a prefix~$\la\lap^k$ and strictly before the 
    prefix~$\la\lap^{k+1}$ ($k\geq 1$).

    \smallskip
    We define for all histories~$h$ ending in a vertex of~$V_j$:
    \[ \sigma_j'(h) = \left\{
    \begin{array}{ll}
      \sigma_j(h) &  \mbox{ if~$\la \not\leq h$,}\\
      \tau_j'(\la\lap^k\delta) & \mbox{ if~$h=\la\delta$.} 
    \end{array}\right. \]
    One can prove that
    $$\pi'=\la\lap^k\pibar' \quad \text{and} \quad \rho'=\la\pibar'.$$
    And then, in the point of view of costs we have
    $$\payoff_j(\rho') < \payoff_j(\pi') < y_j=x_j,$$ 
    which proves that~$\sigma_j'$ is a profitable deviation for player~$j$
    w.r.t.~$(\sigma_i)_{\ipi}$.\qed
  \end{enumerate}
\end{proof}

The next proposition achieves the first step of the proof of 
Theorem~\ref{theo:fin-mem} as mentioned in Section~\ref{sub:fin-mem}.
It shows that one can construct from a Nash equi\-li\-brium another
 Nash equilibrium with the same \type\
and with an outcome of the form~$\abo$. Its proof uses 
Lemmas~\ref{lemma 1} and~\ref{lemma 2}.

\begin{proposition}\label{prop:fin-mem,NE abo}
  Let~$(\sigma_i)_{\ipi}$ be a Nash equilibrium in a
  game~$\mathcal{G}$. Then there exists a Nash
  equilibrium~$(\tau_i)_{\ipi}$ with the same \type\ and such
  that~$\langle (\tau_i)_{\ipi} \rangle = \abo$,
  where~$\visit(\alpha)=\valtype((\sigma_i)_{\ipi})$
  and~$|\alpha\beta| < (|\Pi|+1)\cdot|V|$.
\end{proposition}

\begin{proof}
  Let us set~$\Pi=\{1,\ldots,n\}$ and~$\rho=\langle (\sigma_i)_{\ipi} \rangle$.
  Without loss of generality suppose that 
  \begin{align*}
     \payoff(\rho)=(x_1,\ldots,x_n)  & \ \ \text{ such that } 
    x_1 \leq \ldots \leq x_{\jk} < +\infty\\
     & \ \ \text{ and } x_{\jk+1}=\ldots=x_n=+\infty
  \end{align*}
  where~$0 \leq \jk \leq n$.
  We consider two cases:
  \begin{enumerate}
  \item[$(i)$] $x_1 \geq |V|$.

    \smallskip    
    Then, there exists a prefix~$\la\lap$
    of~$\rho$, with $\lap$ containing at least one vertex, such that 
    \begin{align*}
      &|\la\lap| < x_1\\
      &\visit(\la)=\visit(\la\lap)=\emptyset\\
      &\last(\la)=\last(\la\lap).
    \end{align*}
    We define the strategy profile~$(\tau_i)_{\ipi}$ as proposed in
    Lemma~\ref{lemma 1}. By this lemma it is actually a Nash equilibrium
    in~$\mathcal{G}$. With~$\pi=\langle (\tau_i)_{\ipi} \rangle$,
    we have 
    $$\rho=\la\lap\rhobar \quad \text{and} \quad \pi = \la\rhobar.$$
    Thus if the cost profile for the play~$\pi$ is~$(y_1,\ldots,y_n)$, 
    we have
    \begin{align*}
      & y_1<x_1,\ldots,y_{\jk}<x_{\jk}\\
      & y_{\jk+1}=x_{\jk+1}=+\infty,\ldots,y_{n}=x_{n}=+\infty.
    \end{align*}

  \item[$(ii)$] $(x_{l+1}-x_l) \geq |V|$ for~$1 \leq l \leq \jk-1$.

    \smallskip
    Then, there exists a prefix~$\la\lap$
    of~$\rho$, with $\lap$ containing at least one vertex, such that 
    \begin{align*}
      & x_l <|\la\lap| < x_{l+1}\\
      &\visit(\la)=\visit(\la\lap)=\{1,\ldots,l\}\\
      &\last(\la)=\last(\la\lap).
    \end{align*}
    We define the strategy profile~$(\tau_i)_{\ipi}$ given in
    Lemma~\ref{lemma 1}. It is then a Nash equilibrium
    in~$\mathcal{G}$, and for~$\pi=\langle (\tau_i)_{\ipi} \rangle$, we have 
    $$\rho=\la\lap\rhobar \quad \text{and} \quad \pi = \la\rhobar.$$
    Hence if the cost profile for the play~$\pi$ is~$(y_1,\ldots,y_n)$, 
    we have
    \begin{align*}
    & y_1=x_1,\ldots,y_{l}=x_{l};\\
    & y_{l+1}<x_{l+1},\ldots,y_{\jk}<x_{\jk};\\
    & y_{\jk+1}=x_{\jk+1}=+\infty,\ldots,y_{n}=x_{n}=+\infty.
    \end{align*}
  \end{enumerate}
  By applying finitely many times the two previous cases, we can
  assume without loss of generality that~$(\sigma_i)_{\ipi}$ is a Nash
  equilibrium with a cost profile~$(x_1,\ldots,x_n)$ such that
  \[\begin{array}{ll}
  x_i < i\cdot |V| & \text{\ \ for } i \leq \jk;\\
  x_i = +\infty & \text{\ \ for } i > \jk.
  \end{array}\]

  Let us go further. We can write~$\rho=\alpha\beta\rhobar$ such that
  \begin{align*}
    & \visit(\alpha)=\visit(\rho)\\
    & \last(\alpha)=\last(\alpha\beta)\\
    & |\alpha\beta| < (\jk+1) \cdot |V| \leq (n+1) \cdot |V|.
  \end{align*}
  Indeed, the prefix~$h$ of~$\rho$ of length~$(k+1)\cdot|V|$ visits
  each goal set~$\F_i$, with~$i \leq \jk$, and after the last
  visited~$\F_{\jk}$, there remains enough vertices to observe a
  cycle. Notice that~$\visit(\alpha)=\visit(\alpha\beta)=
  \visit(\rho)$ ($=\valtype((\sigma_i)_{\ipi})$).

  If we define the strategy profile~$(\tau_i)_{\ipi}$ like in
  Lemma~\ref{lemma 2}, we get a Nash equilibrium in~$\mathcal{G}$
  with outcome~$\abo$ and the same \type\ as~$(\sigma_i)_{\ipi}$.
\qed
\end{proof}

We are now ready to prove Theorem~\ref{theo:fin-mem}.
\begin{proof}[of Theorem~\ref{theo:fin-mem}]
  Let us set~$\Pi=\{1,\ldots,n\}$.
  Let~$(\sigma_i)_{\ipi}$ be a Nash equilibrium in the game~$\mathcal{G}$.
  The first step consists in constructing a Nash equilibrium as in
  Proposition~\ref{prop:fin-mem,NE abo}. Let us denote it again 
  by~$(\sigma_i)_{\ipi}$. Let us set~$\rho=\langle (\sigma_i)_{\ipi} \rangle =
  \abo$ such that~$\visit(\alpha)=\valtype((\sigma_i)_{\ipi})$ and~$|\alpha
  \beta| < (n+1)\cdot|V|$. The strategy profile~$(\sigma_i)_{\ipi}$
  is also a Nash equilibrium in the game~$\ginf$ played on the 
  unraveling~$\tinf$ of~$G$.

  For the second step we consider~$\tfin$ the truncated tree of~$\tinf$
  of depth~$\depth=(n+2)\cdot|V|$. It is clear that 
  the strategy profile~$(\sigma_i)_{\ipi}$
  limited to this tree is also a Nash equilibrium of~$\gfin$.

  We know that~$|\alpha\beta| <(n+1)\cdot|V|$ and we set~$\gamma$ such
  that~$\alpha\beta\gamma$ is a prefix of~$\rho$
  and~$|\alpha\beta\gamma|=(n+2)\cdot|V|$.  Furthermore we have
  $\last(\alpha)=\last(\alpha\beta)$ and $\visit(\alpha)=\visit(\alpha
  \beta\gamma)$ (since~$\visit(\alpha)=\valtype(\rho)$).  Then this
  prefix~$\alpha\beta\gamma$ satisfies the properties described in
  Lemma~\ref{lemma:jeu djsn} (by setting~$\nbband=n+1$).  By
  Lemma~\ref{lemma:NE same type} we conclude that there exists a Nash
  equilibrium~$(\tau_i)_{\ipi}$ with finite memory such
  that~$\valtype((\tau_i)_{\ipi})=\visit(\alpha)$, that is, with the
  same \type\ as the initial Nash equilibrium~$(\sigma_i)_{\ipi}$.
  \qed
\end{proof}


\section{Secure Equilibria}\label{sec:SE}




In the previous section, we positively solved Problem 1 and Problem 2
for Nash equilibria. We here solve these two problems for secure
equilibria, but in two-player games only. The main results are stated
in Theorems~\ref{theo:ex ES} and~\ref{theo:fin-mem ES} below. In this
section, we exclusively consider two-player games.

\begin{theorem}\label{theo:ex ES}
  In every quantitative two-player reachability game, there exists a
  finite-memory secure equilibrium.
\end{theorem}

The proof of Theorem~\ref{theo:ex ES} is based on the same ideas as
for the proof of Theorem~\ref{theo:ex EN} (existence of a Nash
equilibrium).  By Kuhn's theorem (Theorem~\ref{theo:kuhn}), there
exists a secure equilibrium in the game~$\gfin$ played on the finite
tree~$\tfin$, for any depth~$\depth$.  By choosing an adequate
depth~$\depth$, Proposition~\ref{prop:ES ssi ES fini} enables to
extend this secure equilibrium to a secure equilibrium in the infinite
tree~$\tinf$, and thus in~$\mathcal{G}$.

The notion of secure equilibrium is based on the binary
relations~$\prec_j$ of Definition~\ref{def:ES}. One can easily see
that~$\prec_j$ is not reflexive.  To be able to apply Kuhn's theorem,
it is more convenient to define secure equilibria via a preference
relation.  Given two cost profiles~$(x_1,x_2)$ and~$(y_1,y_2)$:
$$ (x_1,x_2) \precsim_j (y_1,y_2) \quad \text{ iff } \quad (x_1,x_2)
\prec_j (y_1,y_2) \quad \vee \quad (x_1=y_1 \wedge x_2=y_2)\,.$$ The
relation $\precsim_j$ is clearly a preference relation\footnote{Remark
  that $\precsim_j$ is a kind of lexicographic order on~$(\N \cup
  \{+\infty\}) \times (\N \cup \{+\infty\})$.}.  We can now provide an
equivalent definition of secure equilibrium.


\begin{proposition}\label{prop:def ES}
  A strategy profile~$(\sigma_1,\sigma_2)$ of a game~$\mathcal{G}$ is
  a secure equi\-li\-brium iff for all strategies~$\sigma_1'$ of
  player~$1$ in~$\mathcal{G}$, we have:
  $$\payoff(\rho') \precsim_1 \payoff(\rho)$$
  where~$\rho = \langle \sigma_1,\sigma_2\rangle$ and~$\rho' =
  \langle \sigma_1',\sigma_2 \rangle$, and symmetrically for
  all strategies~$\sigma_2'$ of player~$2$.
\end{proposition}

Since $\precsim_1$ and $\precsim_2$ are preference relations, we get
the next corollary by Kuhn's theorem.

\begin{corollary}\label{coro:kuhn ES}
  Let~$\mathcal{G}$ be a quantitative two-player reachability game
  and~$\tinf$ be the unraveling of~$G$. Let~$\gfin$ be the game played
  on the truncated tree of~$\tinf$ of depth~$\depth$, with~$\depth
  \geq 0$. Then there exists a secure equilibrium in~$\gfin$.
\end{corollary}

Now that we can guarantee the existence of secure equilibrium in
finite trees, it remains to show how to lift them to infinite trees.
The next proposition states that it is possible to extend a secure
equilibrium in~$\gfin$ to a secure equilibrium in the game~$\ginf$
with the same type, if the depth~$d$ is greater or equal to~$(|\Pi| +
1)\cdot 2\cdot|V|$ and there are only two players.  It also says that
we can construct a secure equilibrium in~$\gfin$ from a secure
equilibrium in~$\ginf$, while keeping the same type.

\begin{proposition}\label{prop:ES ssi ES fini}
  Let~$\mathcal{G}$ be a two-player game and $\tinf$ be the unraveling
  of~$G$.
  \begin{enumerate}
  \item[$(i)$] If there exists a secure equilibrium of a certain type in the
    game~$\ginf$, then there exists a secure equilibrium of the same
    type in the game~$\gfin$, for some depth~$\depth \geq (|\Pi| +
    1)\cdot 2\cdot|V|$.
  \item[$(ii)$] If there exists a secure equilibrium of a certain type in the
    game~$\gfin$, where~$\depth \geq (|\Pi| + 1)\cdot 2\cdot|V|$, then
    there exists a finite-memory secure equilibrium of the same type
    in the game~$\ginf$.
  \end{enumerate}
\end{proposition}

To prove Proposition~\ref{prop:ES ssi ES fini}, we need the following
technical lemma whose hypotheses are the same as in
Lemma~\ref{lemma:jeu djsn}.  Recall that Lemma~\ref{lemma:jeu djsn}
states that for all~$j \in \Pi$ such that~$\alpha$ does not visit
$\F_j$, the players~$i\not = j$ can play together to prevent
player~$j$ from reaching his goal set~$\F_j$ from any
history~$h\vertu$ consistent with~$(\sigma_i)_{\ipimj}$ and such
that~$|h\vertu| \leq |\alpha \beta|$.  We denote by~$\strat_{-j}$ the
\memoryless\ winning strategy of the coalition, and for each player~$i
\not= j$, $\strat_{i,j}$ the \memoryless\ strategy of player~$i$
in~$\mathcal{G}$ induced by~$\strat_{-j}$.

\begin{lemma}\label{lemma:visit(alpha)=visit(rho)}
  Suppose~$d \geq 0$. Let~$(\sigma_1,\sigma_2)$ be a secure
  equilibrium in $\gfin$ and $\rho=\langle \sigma_1,\sigma_2 \rangle$
  its outcome.  Assume that $\rho$ has a prefix $\alpha\beta\gamma$,
  where $\beta$ contains at least one vertex, such that
  \begin{align*}
    & \visit(\alpha)=\visit(\alpha\beta\gamma)\\
    & \last(\alpha)=\last(\alpha\beta)\\
    & |\alpha\beta| \leq \nbband\cdot|V|\\
    & |\alpha\beta\gamma| = (\nbband+1)\cdot|V|
  \end{align*}
  for some $\nbband \geq 1$. Then we have 
  $$(\visit(\alpha)\not=\emptyset \vee \visit(\rho) \not= \{1,2\})
  \Rightarrow \visit(\alpha)=\visit(\rho).$$
\end{lemma}

In particular, Lemma~\ref{lemma:visit(alpha)=visit(rho)} implies that
if~$\alpha$ visits none of the goal sets, then $\rho$ visits either
both goal sets or none. Notice that in the case of Nash equilibria, we
can have situations contradicting
Lemma~\ref{lemma:visit(alpha)=visit(rho)}, and in particular the
previous situation, as it can be seen in Example~\ref{ex:diff-types}.

\begin{proof}
  By contradiction, assume that~$2 \in \visit(\rho) \setminus
  \visit(\alpha)$ (the case where $1 \in \visit(\rho) \setminus
  \visit(\alpha)$ is symmetric).  The hypothesis implies that~$1 \in
  \visit(\alpha)$ or~$1 \not \in \visit(\rho)$.

  By Lemma~\ref{lemma:jeu djsn}, player~1 wins the
  game~$\mathcal{G}_2$ from~$\last(\alpha)$, that is, has a
  \memoryless\ winning strategy~$\strat_{1,2}$ from this vertex. Then
  if player~1 plays according to~$\sigma_1$ until depth~$|\alpha|$,
  and then switches to~$\strat_{1,2}$ from~$\last(\alpha)$, this
  strategy is a $\prec_1$-profitable deviation for player~1
  w.r.t.~$(\sigma_1,\sigma_2)$. Indeed, if~$1 \in \visit(\alpha)$,
  player~1 manages to increase player~2's cost while keeping his own
  cost.  On the other hand, if~$1 \not \in \visit(\rho)$, either
  player~1 succeeds in reaching his goal set (i.e. strictly decreasing
  his cost), or he does not reach it (then gets the same cost as
  in~$\rho$) but succeeds in increasing player~2's cost. Thus we get a
  contradiction.  \qed
\end{proof}

We can now give the proof of Proposition~\ref{prop:ES ssi ES fini}.
The idea for showing case $(i)$ is to look at the play~$\pi$ of the
secure equilibrium in~$\ginf$ and consider the depth~$\depth$ needed
to visit all the goal sets of the players in~$\visit(\pi)$. Then, the
secure equilibrium in~$\gfin$ is defined exactly as the secure
equilibrium of~$\ginf$.

The proof of case $(ii)$ works pretty much as the one of
Proposition~\ref{prop:NE fin ds inf} (whereas the latter proposition
does not preserve the type of the Nash equilibrium). Thanks to
Lemma~\ref{lemma:visit(alpha)=visit(rho)}, the proof reduces into only
two cases depending on when the goal sets are visited. In the most
interesting case, a well-chosen prefix~$\alpha\beta$, with~$\beta$
being a cycle, is first extracted from the outcome~$\rho$ of the
secure equilibrium~$(\sigma_1, \sigma_2)$ of~$\gfin$. The outcome of
the required secure equilibrium of~$\ginf$ will be equal to~$\abo$. As
soon as a player deviates from this play, the other player punishes
him, but the way to define the punishment is here more involved than
in the proof of Proposition~\ref{prop:NE fin ds inf}.  In the other
case, the proof is simpler, but the ideas are quite the same.

Before entering the details, let us introduce a notation. For any play
$\rho= \rho_0\rho_1\ldots$ of $\mathcal{G}$ and any player~$i \in
\Pi$, we define $\ind_i(\rho)$ as the least index $l$ such that
$\rho_l \in \F_i$ if it exists, or $-1$ if not\footnote{We are
  conscious that it is counterintuitive to use the particular value
  $-1$, but it is helpful in the proofs.}.

\begin{proof}[of Proposition~\ref{prop:ES ssi ES fini}]
  First let us begin with the proof of~$(i)$. Suppose that there
  exists a secure equilibrium~$(\tau_1,\tau_2)$ in~$\ginf$ and that
  the play~$\pi$ is the outcome of this strategy profile.  Let us set
  $\depth := \max\{(|\Pi| + 1)\cdot 2\cdot|V|,\ind_1(\pi),
  \ind_2(\pi)\}$ and define~$(\sigma_1,\sigma_2)$ as the strategy
  profile in~$\gfin$ corresponding to the strategies $(\tau_1,\tau_2)$
  restricted to the finite tree.  Clearly the outcome~$\rho$
  of~$(\sigma_1,\sigma_2)$ is a prefix of~$\pi$
  and~$\visit(\rho)=\visit(\pi)$, so $(\sigma_1,\sigma_2)$ and
  $(\tau_1,\tau_2)$ are of the same type.  It remains to show
  that~$(\sigma_1,\sigma_2)$ is a secure equilibrium in~$\gfin$.
 
  Assume by contradiction that player~1 has a $\prec_1$-profitable
  deviation~$\sigma_1'$ w.r.t. $(\sigma_1,\sigma_2)$ (the case of
  player~2 is symmetric).  We write~$\rho'$ for the outcome
  of~$(\sigma_1',\sigma_2)$ in $\gfin$.  There are two cases to
  consider: either player~$1$ manages to decrease his cost in~$\rho'$
  w.r.t.~$\rho$, or he pays the same cost as in~$\rho$ but he is able
  to increase the cost of player~$2$ in~$\rho'$ w.r.t.~$\rho$.  In
  both cases, if player~1 plays according to~$\sigma_1'$ in~$\ginf$
  until depth~$\depth$ and then arbitrarily, one can easily be
  convinced that we get a $\prec_1$-profitable
  deviation\footnote{Notice that in the second case, when~$\rho$ does
    not visit~$\F_1$ in~$\gfin$, player~1 may reach his goal set
    in~$\ginf$ when deviating in this way, and this would be
    profitable for him in this game.}  w.r.t.~$(\tau_1,\tau_2)$
  in~$\ginf$.
  This leads to a contradiction. 
  
  Now let us proceed to the proof of~$(ii)$.
  Let~$(\sigma_1,\sigma_2)$ be a secure equilibrium in the
  game~$\gfin$, where~$\depth \geq (|\Pi| + 1)\cdot 2\cdot|V|$,
  and~$\rho$ its outcome.  We define the prefixes~$\la\lap$
  and~$\alpha\beta\gamma$ as in the proof of Proposition~\ref{prop:NE
    fin ds inf} (see Figure~\ref{fig:slicing play}).

  By Lemma~\ref{lemma:visit(alpha)=visit(rho)} there are only two
  cases to consider:
  \begin{enumerate}
  \item[(a)] $\visit(\alpha)=\emptyset$ and $\visit(\rho)=\{1,2\}$;
  \item[(b)] $\visit(\alpha)=\visit(\rho)$.
  \end{enumerate}
  We define a different secure equilibrium according to the case.

  Let us start with case~(a): $\visit(\alpha)=\emptyset$ and
  $\visit(\rho)=\{1,2\}$. We define the following strategy profile:
  \[ \tau_i(h) = \left\{
  \begin{array}{ll}
    \sigma_i(h) & \mbox{ if~$|h| < \max\{\ind_1(\rho),\ind_2(\rho)\}$,}\\
    \textit{arbitrary} & \mbox{ otherwise.}\\
  \end{array}\right.\]
  where~$i=1,2$, and \textit{arbitrary} means that the next vertex is
  chosen arbi\-tra\-rily, but in a \memoryless\ way. 
  Note that the outcome of~$(\tau_1,\tau_2)$ is of the form
  $\alpha'(\beta')^\omega$ where
  $\visit(\alpha')=\visit(\rho)=\{1,2\}$ and $\beta'$ is a cycle.
  So,~$(\tau_1,\tau_2)$ has the same type as~$(\sigma_1,\sigma_2)$.
  It remains to prove that $(\tau_1,\tau_2)$ is a finite-memory secure
  equilibrium in~$\ginf$.

  Assume by contradiction that player~1 has a~$\prec_1$-profitable
  deviation~$\tau_1'$ w.r.t. $(\tau_1,\tau_2)$ in~$\ginf$ (the case
  for player~2 is symmetric).  The strategy~$\sigma_1'$ equal
  to~$\tau_1'$ in $\gfin$ is clearly a~$\prec_1$-profitable deviation
  w.r.t. $(\sigma_1,\sigma_2)$, which is a contradiction with the fact
  that~$(\sigma_1,\sigma_2)$ is a secure equilibrium in
  $\gfin$. 
  Moreover, as done in the proof of Lemma~\ref{lemma:NE same type},
  $(\tau_1,\tau_2)$ is a finite-memory strategy profile.



  Now we consider case (b): $\visit(\alpha)= \visit(\rho)$.  Like in
  the proof of Lemma~\ref{lemma:NE same type} we consider the infinite
  play~$\abo$ in the game~$\ginf$.  The basic idea of the strategy
  profile~$(\tau_1,\tau_2)$ is the same as for the Nash equilibrium
  case: player~$2$ (resp.~1) plays according to~$\abo$ and punishes
  player~$1$ (resp.~2) if he deviates from~$\abo$, in the following
  way. Suppose that player~1 deviates (the case for player~2 is
  similar). Then player~2 plays according to~$\sigma_2$ until
  depth~$|\alpha|$, and after that, he plays arbitrarily if $\alpha$
  visits~$\F_1$, otherwise he plays according to~$\strat_{2,1}$.

  We define the same punishment function~$\pun$ as in the proof of
  Lemma~\ref{lemma:NE same type}: for $v_0$, we define $\pun(v_0) =
  \bot$ and for~$h \in V^+$ such that~$\last(h) \in V_i$ and $v \in
  V$, we let:
  \[ \pun(hv) = \left\{
  \begin{array}{ll}
    \bot & \mbox{ if~$\pun(h)= \bot$ and $hv < \abo$,}\\
    i & \mbox{ if~$\pun(h)= \bot$ and $hv \not< \abo$,}\\
    \pun(h) & \mbox{ otherwise ($\pun(h) \not=\bot$).} 
  \end{array}\right. \]

The definition of the secure equilibrium~$(\tau_1,\tau_2)$ is as
follows: for $h \in H$ such that $\last(h) \in V_i$:
  \[ \tau_i(h) = \left\{
  \begin{array}{ll}
    v
    & \mbox{ if~$\pun(h)=\bot$ ($h<\abo$); such that
     $hv<\abo$,}\\
    \textit{arbitrary} & \mbox{ if~$\pun(h)=i$,}\\
    \sigma_i(h) & \mbox{ if~$\pun(h) \not= \bot,i$ and
     $|h| \leq |\alpha|$,}\\
    \strat_{i,\pun(h)}(h) & \mbox{ if~$\pun(h) \not= \bot,i$, 
     $|h| > |\alpha|$ and $\alpha$ does not visit $\F_{\pun(h)}$,}\\
    \textit{arbitrary} & \mbox{ otherwise ($\pun(h) \not= \bot,i$,
     $|h| > |\alpha|$ and $\alpha$ visits~$\F_{\pun(h)}$)}\\
  \end{array}\right.\]
  where~$i=1,2$, and \textit{arbitrary} means that the next vertex is
  chosen arbi\-tra\-rily (in a \memoryless\ way). Clearly the outcome
  of~$(\tau_1,\tau_2)$ is the play~$\abo$, and the type
  of~$(\tau_1,\tau_2)$ is equal to~$\visit(\alpha)=\visit(\rho)$, the
  type of~$(\sigma_1,\sigma_2)$.  Moreover, as done in the proof of
  Lemma~\ref{lemma:NE same type}, $(\tau_1,\tau_2)$ is a finite-memory
  strategy profile.
  
  Remark that the definition of the strategy profile~$(\tau_1,\tau_2)$
  is a little different from the one in the proof of
  Lemma~\ref{lemma:NE same type} because here, if player~1 deviates
  (for example), then player~2 has to prevent him from reaching his
  goal set~$\F_1$ (faster), or having the same cost but succeeding in
  increasing player~2's cost.

  It remains to show that~$(\tau_1,\tau_2)$ is a secure equilibrium in
  the game~$\ginf$.  Assume by contradiction that there exists
  a~$\prec_1$-profitable deviation~$\tau_1'$ for player~$1$
  w.r.t.~$(\tau_1,\tau_2)$. The case of a~$\prec_2$-profitable
  deviation~$\tau_2'$ for player~$2$ is similar.  We construct a
  play~$\rho'$ in~$\gfin$ as follows: player~$1$ plays according to
  the strategy~$\tau_1'$ restricted to $\gfin$ (denoted
  by~$\sigma_1'$) and player~$2$ plays according to~$\sigma_{2}$.
  Thus the play~$\rho'$ coincide with the play~$\pi'=\langle
  \tau_1',\tau_{2} \rangle$ at least until depth~$|\alpha|$ (by
  definition of~$\tau_{2}$); it can differ afterwards.  We have:
  \[
  \begin{array}{lllll}
    \rho&=&\langle \sigma_1,\sigma_{2}\rangle
    &\text{ of cost profile } & (x_1,x_{2})\\
    \rho'& =&\langle \sigma_1',\sigma_{2} \rangle
    &\text{ of cost profile } & (x_1',x_{2}')\\
    \pi &=&\langle \tau_1,\tau_{2}\rangle
    &\text{ of cost profile } & (y_1,y_{2})\\
    \pi'&=&\langle \tau_1',\tau_{2}\rangle
    &\text{ of cost profile } & (y_1',y_{2}').
  \end{array} \]
  The situation is depicted in Figure~\ref{fig:ES plays}.



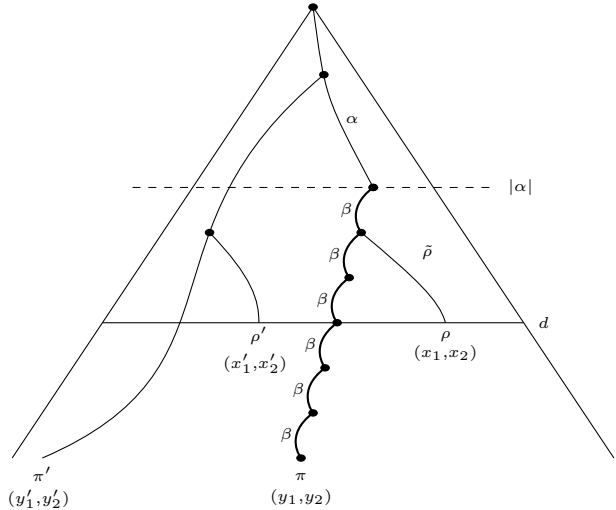
\begin{figure}[h!]
  \centering
  \begin{tikzpicture}[yscale=.6,xscale=.8]
    \everymath{\scriptstyle}
    \draw (0,10) -- (-5,0);
    \draw (0,10) -- (5,0);
    
    \draw[very thin, dashed] (-3,6) -- (3,6);
    \draw (3.1,6) node[right] (q0) {$|\alpha|$};

    \draw (.4,7.5) node[right] (q0) {$\alpha$};

    \draw (-3.5,3) -- (3.5,3);
    \draw (3.6,3) node[right] (q0) {$\depth$};

    \draw[fill=black] (0,10) circle (2pt);
    \draw[] (0,10) .. controls (.2,8) .. (1,6);
    \draw (.8,5.5) node[left] (q0) {$\beta$};
    \draw (.6,4.5) node[left] (q0) {$\beta$};
    \draw (.4,3.5) node[left] (q0) {$\beta$};
    \draw (.2,2.5) node[left] (q0) {$\beta$};
    \draw (0,1.5) node[left] (q0) {$\beta$};
    \draw (-.2,.5) node[left] (q0) {$\beta$};

    \draw[fill=black] (1,6) circle (2pt); 

    \draw[] (.8,5) .. controls (1.7,4) and (2.1,3.5) .. (2.2,3);
    \draw (2.2,3) node[below] (q0) {$\rho$};
    \draw (2.2,2.7) node[below] (q0) {$(x_1,x_2)$};

    \draw[fill=black] (-1.72,5) circle (2pt);
    \draw[] (-1.72,5) .. controls (-1,4) and (-.9,3.5) .. (-.9,3);
    \draw (-.9,3.1) node[below] (q0) {$\rho'$};
    \draw (-.96,2.5) node[below] (q0) {$(x'_1,x'_2)$};

    \draw (1.9,4.9) node[below] (q0) {$\tilde{\rho}$};

    \draw[fill=black] (1,6) circle (2pt);
    \draw[thick] (1,6) .. controls (.5,5.5) and (.8,5) .. (.8,5);
    
    \begin{scope}[yshift=-1cm,xshift=.8cm]
      \draw[fill=black] (0,6) circle (2pt);
      \draw[thick] (0,6) .. controls (-.5,5.5) and (-.2,5) .. (-.2,5);
    \end{scope}
    \begin{scope}[yshift=-2cm,xshift=.6cm]
      \draw[fill=black] (0,6) circle (2pt);
      \draw[thick] (0,6) .. controls (-.5,5.5) and (-.2,5) .. (-.2,5);
    \end{scope}
    \begin{scope}[yshift=-3cm,xshift=.4cm]
      \draw[fill=black] (0,6) circle (2pt);
      \draw[thick] (0,6) .. controls (-.5,5.5) and (-.2,5) .. (-.2,5);
    \end{scope}
    \begin{scope}[yshift=-4cm,xshift=.2cm]
      \draw[fill=black] (0,6) circle (2pt);
      \draw[thick] (0,6) .. controls (-.5,5.5) and (-.2,5) .. (-.2,5);
    \end{scope}
    \begin{scope}[yshift=-5cm]
      \draw[fill=black] (0,6) circle (2pt);
      \draw[thick] (0,6) .. controls (-.5,5.5) and (-.2,5) .. (-.2,5);
    \end{scope}
    \begin{scope}[yshift=-6cm,xshift=-.2cm]
      \draw[fill=black] (0,6) circle (2pt);
    \end{scope}

    \draw (-.2,-.15) node[below] (q0) {$\pi$};
    \draw (-.2,-.5) node[below] (q0) {$(y_1,y_2)$};

    \draw[fill=black] (.18,8.5) circle (2pt);
    \draw[] (.18,8.5) .. controls (-3,5) and (-1,2) .. (-4.5,0);
    \draw (-4.5,0) node[below] (q0) {$\pi'$};
    \draw (-4.56,-.5) node[below] (q0) {$(y'_1,y'_2)$};

    
    
  \end{tikzpicture}
  \caption{Plays~$\rho$ and~$\pi$, and their respective deviations~$\rho'$
    and~$\pi'$.}
  \label{fig:ES plays}
\end{figure}


  By contradiction, we assumed that $\tau_1'$ is a
  $\prec_1$-profitable deviation for player~1
  w.r.t. $(\tau_1,\tau_2)$, i.e. $(y_1,y_2) \prec_1 (y_1',y_2')$. Now
  we are going to show that~$(x_1,x_{2}) \prec_1 (x_1',x_{2}')$,
  meaning that~$\sigma_1'$ is a~$\prec_1$-profitable deviation for
  player~$1$ w.r.t.~$(\sigma_1,\sigma_{2})$ in~$\gfin$. This will lead
  to the contradiction. As~$\tau_1'$ is a~$\prec_1$-profitable
  deviation w.r.t.~$(\tau_1,\tau_2)$, one of the following three cases
  stands.
  \begin{enumerate}
  \item[(1)] $y_1' < y_1 < +\infty$. 

    \smallskip
    As~$\pi=\abo$, it means that~$\alpha$ visits~$F_1$, and then:
    $$y_1'<y_1 = x_1 \leq |\alpha|.~$$ As~$y_1'<|\alpha|$, we
    have~$x_1'=y_1'$ (as~$\rho'$ and~$\pi'$ coincide until
    depth~$|\alpha|$). Therefore $x_1' < x_1$, and $(x_1,x_{2})
    \prec_1 (x_1',x_{2}')$.
    \smallskip
  \item[(2)] $y_1' < y_1 = +\infty$. 
    
    \smallskip
    If~$y_1' \leq |\alpha|$, we have~$x_1'=y_1'$ (by the same argument
    as before).  As~$\visit(\alpha)=\visit(\rho)$, we have
    $x_1=y_1=+\infty$ and $x_1' < x_1$ (and so $(x_1,x_{2}) \prec_1
    (x_1',x_{2}')$).

    We show that the case~$y_1' > |\alpha|$ is impossible.  By
    definition of~$\tau_{2}$ the play~$\pi'$ is consistent
    with~$\sigma_{2}$ until depth~$|\alpha|$, and then
    with~$\strat_{2,1}$ (as~$y_1=+\infty$). By Lemma~\ref{lemma:jeu
      djsn} the play~$\pi'$ can not visit~$\F_1$ after a depth~$>
    |\alpha|$.
    \smallskip
  \item[(3)] $y_1=y_1'$ and~$y_{2} < y_{2}'$. 
    
    \smallskip
    Note that this implies~$y_{2} <+\infty$ and~$x_{2}=y_{2}$
    (as~$\pi=\abo$).  Since~$\rho'$ and~$\pi'$ coincide until
    depth~$|\alpha|$, $y_2 < y_2'$ and~$x_2=y_2\leq |\alpha|$, we have
    $$x_{2}=y_{2} < x_{2}'$$ showing that the cost of player~2 is
    increased.  In order to ensure that $\sigma_1'$ is a
    $\prec_1$-profitable deviation, it remains to show that either
    player~$1$ keeps the same cost, or he decreases his cost.

    If~$y_1'=y_1<+\infty$, it follows as in the first case that:
    $$y_1 =x_1 \leq |\alpha| \quad \text{and} \quad x_1'= y_1'.$$
    Therefore~$x_1=x_1'$, i.e. player~$1$ has the same cost in~$\rho$
    and~$\rho'$. And so, $(x_1,x_{2}) \prec_1 (x_1',x_{2}')$.

    On the contrary, if~$y_1'=y_1=+\infty$, it follows
    that~$x_1=+\infty$ (as~$\visit(\alpha)=\visit(\rho)$). And so, we
    have that $x_1' < +\infty = x_1$, or $x_1'=x_1$. But in both
    cases, it holds that $(x_1,x_{2}) \prec_1 (x_1',x_{2}')$.
  \end{enumerate}
  In conclusion, we constructed a~$\prec_1$-profitable deviation
  $\sigma_1'$ w.r.t.~$(\sigma_1,\sigma_2)$ in $\gfin$, and then we get
  a contradiction.  \qed
\end{proof}

\begin{remark}
  Let us notice that in case~$(i)$ of Proposition~\ref{prop:ES ssi ES
    fini}, the proof remains valid if we take $\depth =
  \max\{0,\ind_1(\pi),\ind_2(\pi)\}$. Thus, in the statement of
  case~$(i)$, the constraint $\depth \geq (|\Pi| + 1)\cdot 2\cdot|V|$
  can be replaced by $\depth \ge 0$.
\end{remark}



We can now proceed to the proof of Theorem~\ref{theo:ex ES}.

\begin{proof}[of Theorem~\ref{theo:ex ES}]
  Let us set $\depth:=(\Pi+1)\cdot 2 \cdot |V|$ and apply
  Corollary~\ref{coro:kuhn ES} on the game~$\gfin$.  Then we get a
  secure equilibrium in this game. By Proposition~\ref{prop:ES ssi ES
    fini} there exists in~$\mathcal{G}$ a finite-memory secure
  equilibrium with the same type.  \qed
\end{proof}

Theorem~\ref{theo:ex ES} positively answers to Problem~\ref{pbm 1} for
secure equilibria in \emph{two-player} games. The next theorem solves
Problem~\ref{pbm 2} for the same kind of games.

\begin{theorem}\label{theo:fin-mem ES}
  If there exists a secure equilibrium in a quantitative two-player
  reachability game~$\mathcal{G}$, then there exists a finite-memory
  secure equilibrium of the same \type\ in~$\mathcal{G}$.
\end{theorem}

\begin{proof}
  Let~$(\sigma_1,\sigma_2)$ be a secure equilibrium in~$\mathcal{G}$.
  By the first part of Proposition~\ref{prop:ES ssi ES fini}, there
  exists a secure equilibrium of the same type in the game~$\gfin$,
  for a certain depth~$\depth \geq (\Pi+1)\cdot 2 \cdot |V|$.  If we
  apply the second part of Proposition~\ref{prop:ES ssi ES fini}, we
  get a finite-memory secure equilibrium of the same type
  as~$(\sigma_1,\sigma_2)$ in~$\mathcal{G}$.  \qed
\end{proof}

The proof of Theorem~\ref{theo:fin-mem ES} is based on
Proposition~\ref{prop:ES ssi ES fini} which, roughly speaking, ensures
that every secure equilibrium of $\gfin$ can be lifted to a secure
equilibrium \emph{of the same type} in $\ginf$, and conversely. Notice
that Proposition~\ref{prop:ES ssi ES fini} has no counterpart for Nash
equilibria, since we can not guarantee that the type can be preserved,
as it can be seen from Example~\ref{ex:diff-types}.  This approach
makes the proof of Theorem~\ref{theo:fin-mem ES} rather different than
the proof of Theorem~\ref{theo:fin-mem}.

Notice that Proposition~\ref{prop:ES ssi ES fini} stands for
two-player games because its proof uses
Lemma~\ref{lemma:visit(alpha)=visit(rho)} that has been proved only
for two players.

\section{Extensions of the Model}\label{sec:ext}

\subsection{Safety Objectives}\label{sub:safety}

Let us now consider quantitative games played on a graph where some
players have \emph{reachability objectives}, whereas others have
\emph{safety objectives}. As previously, the players with 
reachability objectives want to reach their goal set as soon as
possible. The players with safety objectives want to avoid their
\emph{bad set} or, if impossible, delay its visit as long as
possible. Let us make that precise through the following definition.

\begin{definition}\label{def-game-safety}
  An \emph{infinite turn-based quantitative multiplayer
    reachability/safety game} is a tuple $\safegame$ where
  \begin{itemize}
  \item[\textbullet] $\Pi$ is a finite set of players partitioned into
    $\Pi_r$ and $\Pi_s$ which are the players with reachability and
    safety objectives respectively,
  \item[\textbullet] $\G$ is a finite directed graph 
    where $V$ is the set of vertices, $(V_i)_{i \in \Pi}$ is a partition of $V$
    into the state sets of each player, $v_0 \in V$ is the initial vertex,
    and $E \subseteq V \times V$ is the set of edges, and
  \item[\textbullet] $\F_i \subseteq V$ is the goal set of player~$i$,
    for $i \in \Pi_r$; $\B_i \subseteq V$ is the bad set of
    player~$i$, for $i \in \Pi_s$.
  \end{itemize}
\end{definition}

For any play~$\rho=\rho_0\rho_1\ldots$ of $\mathcal{G}$, we note
$\payoff_i(\rho)$ the \emph{cost} of player~$i$. For $i \in \Pi_r$ the
cost is defined as before and for $i \in \Pi_s$ the cost is defined
by:
\[ \payoff_i(\rho) = \left\{
\begin{array}{ll}
  - l &
  \mbox{ if $l$ is the \emph{least} index such that $\rho_l \in \B_i$,}\\
  - \infty & \mbox{ otherwise.} 
\end{array}\right. \]
As before, the aim of each player~$i$ is to \emph{minimize} his cost,
i.e. reach his goal set~$\F_i$ as soon as possible for $i \in \Pi_r$,
or delay the visit of $\B_i$ as long as possible for $i \in \Pi_s$.
The notions of play, strategy, outcome and Nash equilibrium extend in
a natural way. The main result of this subsection is the following
theorem which solves Problem~\ref{pbm 1} in this framework.



\begin{theorem}\label{theo:ex ENSafe}
  In every quantitative multiplayer reachability/safety game, there
  exists a finite-memory Nash equilibrium.
\end{theorem}

In order to prove Theorem~\ref{theo:ex ENSafe}, we have to revisit the
results of Section~\ref{sec:NE}. Let us first notice that
Lemma~\ref{lemma:jeu djsn} remains true in this framework when
player~$j$ belongs to $\Pi_r$. Lemma~\ref{lemma:NE same type} remains
true, however we have to slightly adapt its proof.

\begin{proof}[of Lemma~\ref{lemma:NE same type} in the case of
  reachability/safety objectives]

  Let us first introduce some notations. In the rest of the proof, we
  denote by $\Pi_r^f$ (resp. $\Pi_s^f$) the subset of players $i \in
  \Pi_r$ (resp. $i \in \Pi_s$) such that $\alpha$ visits $\F_i$
  (resp. $\B_i$) and by $\Pi_r^\infty$ (resp.  $\Pi_s^\infty$) the set
  $\Pi_r \setminus \Pi_r^f$ (resp. $\Pi_s \setminus \Pi_s^f$).

  The punishment function~$\pun$ is defined exactly as in the proof of
  Lemma~\ref{lemma:NE same type}. For $v_0$, we define $\pun(v_0) =
  \bot$ and for~$h \in V^+$ such that~$\last(h) \in V_i$ and $v \in
  V$, we let:
  \[ \pun(hv) = \left\{
  \begin{array}{ll}
    \bot & \mbox{ if~$\pun(h)= \bot$ and~$hv < \abo$,}\\
    i & \mbox{ if~$\pun(h)= \bot$ and~$hv \not< \abo$,}\\
    \pun(h) & \mbox{ otherwise ($\pun(h) \not=\bot$)\,.} 
  \end{array}\right. \]
  
The difference with the proof of Lemma~\ref{lemma:NE same type} arises
in the definition of the Nash equilibrium~$(\tau_i)_{\ipi}$. The new
equilibrium needs to incorporate an adequate punishment for the
players with safety objectives. More precisely, in order to dissuade a
player $j \in \Pi_s^f$ from deviating, the other players punish him by
playing the strategies $(\sigma_i)_{i \in \Pi \setminus \{j\}}$ in
$\tfin$. Notice that a player $j \in \Pi_s^\infty$ has no incentive to
deviate. Formally we define the Nash equilibrium $(\tau_i)_{i \in
  \Pi}$ as follows. For~$h \in H$ such that~$\last(h) \in V_i$,
  \begin{align}\label{eq:def taui safe}   \tau_i(h) = \left\{
  \begin{array}{ll}
    v
    & \mbox{ if~$\pun(h)=\bot$ ($h<\abo$); such that~$hv<\abo$,}\\
    \textit{arbitrary} & \mbox{ if~$\pun(h)=i$,}\\
    \strat_{i,\pun(h)}(h) & \mbox{ if~$\pun(h) \not= \bot,i$ 
      and~$\pun(h) \in \Pi_r^\infty$,}\\
    \sigma_i(h) & \mbox{ if~$\pun(h) \not= \bot,i$,
     $\pun(h) \in \Pi_r^f \cup \Pi_s^f$ and
     $|h| < \depth$,}\\
    \textit{arbitrary} & \mbox{ otherwise,} 
  \end{array}\right.\end{align}
  where \textit{arbitrary} means that the next vertex is chosen arbitrarily
  (in a \memoryless\ way). Clearly the outcome of~$(\tau_i)_{\ipi}$ is
  the play~$\abo$, and~$\valtype((\tau_i)_{\ipi})$ is equal
  to~$\visit(\alpha)$ ($=\visit(\alpha\beta)$).
  
  It remains to prove that~$(\tau_i)_{\ipi}$ is a finite-memory Nash
  equilibrium in the game~$\ginf$. In order to do so, we prove that
  none of the players has a profitable deviation.  For players with
  reachability objectives, the arguments are exactly the same as the
  ones provided in the proof of Lemma~\ref{lemma:NE same type}. Let us
  now consider players with safety objectives. In the case where $j
  \in \Pi_s^\infty$, player~$j$ has clearly no incentive to
  deviate. In the case where $j \in \Pi_s^f$, to decrease his cost,
  player~$j$ has no incentive to deviate after the
  prefix~$\alpha$. Thus we assume that the strategy~$\tau_j'$ causes a
  deviation from a vertex visited in~$\alpha$. By
  Equation~\eqref{eq:def taui safe} the other players first play
  according to~$(\sigma_i)_{\ipimj}$ in~$\gfin$, and then in an
  arbitrary way.

  Suppose that~$\tau_j'$ is a profitable deviation for player~$j$
  w.r.t.~$(\tau_i)_{\ipi}$ in the game~$\ginf$.  Let us set~$\playtau=
  \langle (\tau_i)_{\ipi} \rangle$ and~$\playdev= \langle
  \tau_j',(\tau_i)_{\ipimj} \rangle$. Then
    $$\payoff_j(\playdev) < \payoff_j(\playtau).$$
    On the other hand we know that
    $$\payoff_j(\playtau) = \payoff_j(\rho) \leq |\alpha|.$$   
    So if we limit the play~$\playdev$ in~$\ginf$ to its prefix of
    length~$\depth$, we get a play~$\rho'$ in~$\gfin$ such that
    $$\payoff_j(\rho') \le \payoff_j(\playdev) < \payoff_j(\rho).$$
    Notice that we do not necessarily have that $\payoff_j(\rho') =
    \payoff_j(\playdev)$ (as in the proof of Lemma~\ref{lemma:NE same
      type}) since the bad set $\B_j$ can be visited by $\pi'$ and not
    by $\rho'$. As the play~$\rho'$ is consistent with the
    strategies~$(\sigma_i)_{\ipimj}$ by Equation~\eqref{eq:def taui
      safe}, the strategy~$\tau_j'$ restricted to the tree~$\tfin$ is
    a profitable deviation for player~$j$ w.r.t.~$(\sigma_i)_{\ipi}$
    in the game~$\gfin$. This is impossible.  Moreover, as done in the
    proof of Lemma~\ref{lemma:NE same type}, $(\tau_1,\tau_2)$ is a
    finite-memory strategy profile.

\qed
\end{proof}

Since Lemma~\ref{lemma:NE same type} holds in the context of
reachability/safety objectives, Proposition~\ref{prop:NE fin ds inf}
ensures that the equilibrium in $\gfin$ provided by Kuhn's theorem
(Corollary~\ref{coro:kuhn}) can be lifted to $\ginf$. This proves
Theorem~\ref{theo:ex ENSafe}.

\subsection{Tuples of Costs on Edges}

In this subsection, we come back to a pure reachability framework and
we extend our model in the following way: we assume that edges are
labelled with tuples of positive costs (one cost for each player).
Here we do not only count the number of edges to reach the goal of a
player, but we sum up his costs along the path until his goal is
reached. His aim is still to minimize his global cost for a play. We
generalize Definition~\ref{def-game}.

\begin{definition}\label{def-game-costs}
  An \emph{infinite turn-based quantitative multiplayer reachability
    game with tuples of costs on edges} is a tuple $\gamecosts$ where
  \begin{itemize}
  \item[\textbullet] $\Pi$ is a finite set of players,
  \item[\textbullet] $\G$ is a finite directed graph 
    where $V$ is the set of vertices, $(V_i)_{i \in \Pi}$ is a partition of $V$
    into the state sets of each player, $v_0 \in V$ is the initial vertex,
    and $E \subseteq V \times V$ is the set of edges, 
  \item[\textbullet] $\cost_i: E \to \R^{>0}$ is the cost function of player~$i$
    defined on the edges of the graph,
  \item[\textbullet] $\F_i \subseteq V$ is the goal set of player~$i$.
  \end{itemize}
\end{definition}



We also positively solve Problem~\ref{pbm 1}
for Nash equilibria in this context.

\begin{theorem}\label{theo:ex EN costs}
  In every quantitative multiplayer reachability game with tuples of
  costs on edges, there exists a finite-memory Nash equilibrium.
\end{theorem}

To prove Theorem~\ref{theo:ex EN costs}, we follow the same scheme as
in Section~\ref{sec:NE}. In particular, we rely on Kuhn's theorem
(Corollary~\ref{coro:kuhn}) and need to prove a counterpart of
Lemma~\ref{lemma:jeu djsn}, Lemma~\ref{lemma:NE same type} and
Proposition~\ref{prop:NE fin ds inf} in this framework.

Let us first introduce some notations that will be useful in this
context.  We define $\cmin := \min_{\ipi} \min_{e \in E} \cost_i(e)$,
$\cmax := \max_{\ipi} \max_{e \in E} \cost_i(e)$ and $\K:= \left\lceil
\frac{\cmax}{\cmin} \right\rceil$. It is clear that $\cmin, \cmax >0$ and
$\K \geq 1$.

We also adapt the definition of $\payoff_i(\rho)$, the \emph{cost} 
of player~$i$ for a play~$\rho=\rho_0\rho_1\ldots$,
\[ \payoff_i(\rho) = \left\{
\begin{array}{ll}
  \displaystyle\sum_{k=1}^{k=l} \cost_i((\rho_{k-1},\rho_k)) &
  \mbox{ if $l$ is the \emph{least} index such that $\rho_l \in
    \F_i$,}\\
  + \infty & \mbox{ otherwise.} 
\end{array}\right. \]

The counterpart of Lemma~\ref{lemma:jeu djsn} is the following
one, taking into account the constant $\K$ defined before.

\begin{lemma}\label{lemma:jeu djsn costs}
  Suppose~$d \geq 0$. Let~$(\sigma_i)_{\ipi}$ be a Nash equilibrium
  in~$\gfin$ and~$\rho$ the outcome of~$(\sigma_i)_{\ipi}$.  Assume
  that~$\rho$ has a prefix~$\alpha\beta\gamma$, where~$\beta$ contains
  at least one vertex, such that
  \begin{align*}
    & \visit(\alpha)=\visit(\alpha\beta\gamma)\\
    & \last(\alpha)=\last(\alpha\beta)\\
    & |\alpha\beta| \leq \nbband\cdot|V|\\
    & |\alpha\beta\gamma| = (\nbband+\K)\cdot|V|
  \end{align*}
  for some~$\nbband \geq 1$.\\
  \noindent  Let~$j \in \Pi$ be such that~$\alpha$ does not visit~$\F_j$. 
  Consider the qualitative two-player zero-sum game~$\gamej$. 
  Then for all histories~$h\vertu$ of~$\mathcal{G}$ consistent 
  with~$(\sigma_i)_{\ipimj}$ and such that $h\vertu \leq \alpha\beta$,
  the coalition of the players~$i \not=j$ wins the game~$\mathcal{G}_j$ 
  from~$\vertu$.
\end{lemma}

\begin{proof}[Sketch]
  As for the proof of Lemma~\ref{lemma:jeu djsn} we proceed by
  contradiction and define a play $\rho'$ in the very same way. We
  can deduce that
  \begin{align*}
    \ind_j(\rho') & \leq  |h\vertu| + |V| &
    \text{(by Proposition~\ref{prop:gamei})} \\
    & \leq  (\nbband+1)\cdot|V| & \text{(by hypothesis)} \\
    & \leq  (\nbband+\K)\cdot|V| & \text{(as $\K \geq 1$)} \\
    & \leq  \depth & \text{(as~$\alpha\beta\gamma\leq \rho$).}
  \end{align*}

  The case where $\payoff_j(\rho)=+\infty$ is solved in the same way.
  For the other case $\payoff_j(\rho)<+\infty$, we note
  ${\sf{c}}_j(h\vertu)$ the sum of the costs of player $j$ along the
  prefix $h\vertu$. We have the following inequalities (see
  Figure~\ref{fig:cout}):
  \begin{align*}
    \payoff_j(\rho') & \leq {\sf{c}}_j(h\vertu) + \cmax \cdot |V| &
    \\
    \payoff_j(\rho)\ & > {\sf{c}}_j(h\vertu) + \cmin \cdot \K \cdot |V| & 
    \text{ (as $\ind_j(\rho) > (l+\K)\cdot |V|$)}\\
    &  \geq {\sf{c}}_j(h\vertu) + \cmin \cdot \frac{\cmax}{\cmin} \cdot |V| &
    \text{ (by definition of $\K$)}\\
    & = {\sf{c}}_j(h\vertu) + \cmax \cdot |V|\,.  &
  \end{align*}


\begin{figure}[h!]
    \centering
    \begin{tikzpicture}[yscale=.6,xscale=.75]
      \everymath{\scriptstyle}
      \draw (0,8) -- (-5,0);
      \draw (0,8) -- (5,0);

      \draw[fill=black] (0,8) circle (2pt);

      \draw (-5,0) -- (5,0);

      \draw[very thin,dotted] (-4,1.1) -- (6.1,1.1);
      \draw[very thin,dotted] (-4,3) -- (6.1,3);
      \draw[very thin,dotted] (-4,4) -- (6.1,4);

      \draw[<->] (5.4,1.1)--(5.4,3) node[left,pos=.5] {$\K\cdot|V|$};

      \path (6.1,4) node[right] (q0) {$(\nbband-1)\cdot|V|$};
      \path (6.1,3) node[right] (q0) {$\nbband\cdot|V|$};
      \path (6.1,1.5) node[right] (q0) {$(\nbband+\K)\cdot|V|$};
      \path (6.2,0) node[right] (q0) {$\depth$};

      \path (.2,6) node[left] (q0) {$h$};
      \path (1.2,3.5) node[right] (q0) {$\vertu$};

      \draw[thick] (0,8) .. controls (0.2,5.6) .. (1.2,3.5);
      \draw[fill=black] (1.33,3.2) circle (1pt);
      \draw[fill=black] (1.06,3.8) circle (1pt);
      \draw[fill=black] (1.2,3.5) circle (2pt);
      \draw[very thin] (1.2,3.5) .. controls (1.6,2.6) .. (2,0);

      \draw[very thin] (1.2,3.5) .. controls (-.6,2.1) .. (-1,0);


      \draw [snake=brace] (0,2.75) -- (1,3.55) node [above,pos=.5]
            {{\footnotesize $\le |V| \ \quad {}$}};

      \draw[fill=black] (0,2.5) circle (2pt);
      \draw[fill=black] (1.89,.7) circle (2pt);
      \path (0,2.5) node[left] (q0) {$\F_j$};
      \path (1.89,.7) node[right] (q0) {$\F_j$};



      \draw (2,-.3) node[below] (q0) {$\rho$};
      \draw (-1,-.1) node[below] (q0) {$\rho'$};


    \end{tikzpicture}
    \caption{Plays $\rho$ and $ \rho'$ with their common prefix
      $h\vertu$.}
    \label{fig:cout}
\end{figure}
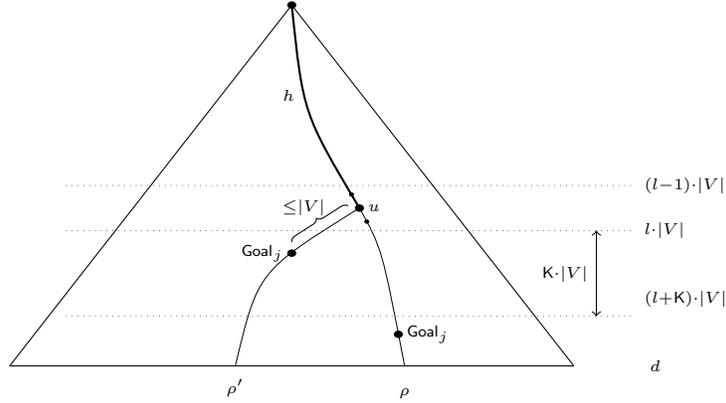

 Then we have $\payoff_j(\rho') < \payoff_j(\rho)$,
  and since~$\rho'$ is consistent with~$(\sigma_i)_{\ipimj}$, the
  strategy of player~$j$ induced by the play~$\rho'$ is a profitable
  deviation for player~$j$ w.r.t.~$(\sigma_i)_{\ipi}$. This
  contradicts the fact that $(\sigma_i)_{\ipi}$ is a Nash equilibrium
  in the game $\gfin$. \qed
\end{proof}

The following lemma is the counterpart of Lemma~\ref{lemma:NE same type}.

\begin{lemma}\label{lemma:NE same type costs}
  Suppose~$\depth \geq 0$.
  Let~$(\sigma_i)_{\ipi}$ be a Nash equilibrium in~$\gfin$ 
  and~$\alpha\beta\gamma$ be a prefix of~$\rho=\langle (\sigma_i)_{\ipi}
  \rangle$ as defined in Lemma~\ref{lemma:jeu djsn costs} where
  $|\alpha\beta\gamma| = (\nbband+\K)\cdot|V|$
  for some~$\nbband \geq 1$ such that
  $l \leq \frac{d}{|V|\cdot \K}+1$.\\ 
  \noindent Then there exists a Nash equilibrium~$(\tau_i)_{\ipi}$ 
  in the game~$\ginf$.
  Moreover~$(\tau_i)_{\ipi}$ is finite-memory,
  and $\valtype((\tau_i)_{\ipi})=\visit(\alpha)$.
\end{lemma}

\begin{proof}
  We prove this result in the very same way as Lemma~\ref{lemma:NE
    same type}.  The only difference lies in the case\footnote{Indeed
    when $j>k$, i.e. when player~$j$ has not reached his goal set, the
    coalition punishes him in the exact same way as
    Lemma~\ref{lemma:NE same type} by preventing him from visiting his
    goal set.} $j \leq k$ when we show that $(\tau_i)_{\ipi}$ is a
  Nash equilibrium.  We suppose that~$\tau_j'$ is a profitable
  deviation for player~$j$ w.r.t.~$(\tau_i)_{\ipi}$ in the
  game~$\ginf$. So we have $\payoff_j(\playdev)
  < \payoff_j(\playtau)$, where~$\playtau= \langle (\tau_i)_{\ipi}
  \rangle$ and~$\playdev= \langle \tau_j',(\tau_i)_{\ipimj}
  \rangle$. As $\ind_j(\playtau) \leq |\alpha|$, we know that
  $\payoff_j(\playtau) \leq |\alpha|\cdot \cmax$. It follows that
  $\payoff_j(\playdev) < |\alpha|\cdot \cmax$ and
  \begin{align*}
    \ind_j(\playdev) & <  |\alpha|\cdot \frac{\cmax}{\cmin} & 
    \\
    & \leq (l-1) \cdot |V| \cdot \K  & \\
    & \leq \depth & \text{(by hypothesis).}
  \end{align*}
  The first inequality can be justified as follows. For a
  contradiction, let us assume that $\ind_j(\playdev) \ge |\alpha|\cdot
  \frac{\cmax}{\cmin}$. It follows that $\payoff_j(\playdev) \ge \cmin
  \cdot |\alpha|\cdot \frac{\cmax}{\cmin}$, this contradicts the fact
  that $\payoff_j(\playdev) < |\alpha|\cdot \cmax$.

  As in the proof of Lemma~\ref{lemma:NE same type}, we limit the
  play~$\playdev$ in~$\ginf$ to its prefix of length~$\depth$ and get
  a profitable deviation for player~$j$ w.r.t.~$(\sigma_i)_{\ipi}$ in
  the game~$\gfin$, contradicting the fact that $(\sigma_i)_{\ipi}$ is
  a Nash equilibrium in $\gfin$.

  Moreover, as done in the proof of Lemma~\ref{lemma:NE same type},
  $(\tau_1,\tau_2)$ is a finite-memory strategy profile.  \qed
\end{proof}

As a consequence of the two previous lemmas, Proposition~\ref{prop:NE
  fin ds inf} remains true in this context, we only have to adjust the
depth~$\depth$ of the finite tree.

\begin{proposition}\label{prop:NE fin ds inf costs}
  Let~$\mathcal{G}$ be a game and
  $\tinf$ be the unraveling of~$G$. Let~$\gfin$ be
  the game played on the truncated tree of~$\tinf$ of 
  depth~$\depth=\max\{(|\Pi| + 1)\cdot (\K+1) \cdot |V|,  
    (|\Pi|\cdot (\K+ 1)+1)\cdot |V| \cdot \K\}$.\\
  \noindent If there exists a Nash equilibrium in the game~$\gfin$,
  then there exists a finite-memory Nash equilibrium in the game~$\ginf$.
\end{proposition}

\begin{proof}
  The proof is similar to the proof of Proposition~\ref{prop:NE fin ds inf}.
  Let~$(\sigma_i)_{\ipi}$ be a
  Nash equilibrium in the game~$\gfin$ and~$\rho$ its outcome.
  We consider the prefix~$\la\lap$ of~$\rho$ of minimal length such that
  \begin{align*}
    \exists \, \nbband \geq 1 \ \ \ \ \ \ & 
    |\la| = (\nbband-1)\cdot|V|\notag\\
    & |\la\lap|=(\nbband+\K)\cdot|V|\notag\\
    & \visit(\la)=\visit(\la\lap). 
  \end{align*}
  In the worst case, 
  the play~$\rho$ visits the goal set of a new player
  in each prefix of length~$i\cdot (\K+1) \cdot|V|$,~$1 \leq i \leq |\Pi|$, 
  i.e.~$|\la|= |\Pi| \cdot (\K+1) \cdot|V|$. So we know 
  that~$\nbband \leq |\Pi| \cdot (\K+1) +1$ and~$\la\lap$ exists as a prefix
  of~$\rho$, because the length~$\depth$ of~$\rho$ is greater or equal
  to~$(|\Pi|+1)\cdot (\K+1) \cdot|V|$ by hypothesis.

  Given the length of~$\lap$ ($\K \geq 1$), one vertex of~$V$ is
  visited at least twice by~$\lap$. More precisely, we can write
  \begin{align*}
    \la\lap = \alpha\beta\gamma  \text{\ \ \  with\ \ \ }
    & \last(\alpha)= \last(\alpha\beta)\notag\\
    & |\alpha| \geq (\nbband -1) \cdot |V| \notag\\
    & |\alpha\beta| \leq \nbband\cdot|V|.
  \end{align*}
  We have~$\visit(\alpha)=\visit(\alpha\beta\gamma)$, and
  $|\alpha\beta\gamma| =(\nbband+\K)\cdot|V|$.

  Moreover, the following inequality holds:
  $$\depth \geq  (|\Pi|\cdot (\K+ 1)+1)\cdot |V| \cdot \K
  \geq \nbband \cdot |V| \cdot \K \quad \text{and so,} \quad \nbband
  \leq \frac{d}{|V|\cdot\K}.$$ Then, we can apply Lemma~\ref{lemma:NE
    same type costs} and get a finite-memory Nash
  equilibrium~$(\tau_i)_{\ipi}$ in the game~$\ginf$ such
  that~$\valtype((\tau_i)_{\ipi})=\visit(\alpha)$.  \qed
\end{proof}

Thanks to Corollary~\ref{coro:kuhn} and Proposition~\ref{prop:NE fin
  ds inf costs}, one can easily deduce Theorem~\ref{theo:ex EN costs}.

Let us comment on the depth~$\depth$ chosen in
Proposition~\ref{prop:NE fin ds inf costs}. It is defined as the
maximum between~$d_1 := (|\Pi| + 1)\cdot (\K+1) \cdot |V|$ and~$d_2:=
(|\Pi|\cdot (\K+ 1)+1)\cdot |V| \cdot \K$. One can easily prove that
$d_1 < d_2$ if and only if $\K^2 > \frac{|\Pi|+1}{|\Pi|}$.

\smallskip

We now investigate an alternative method to handle \emph{simple cost
  functions}. More precisely we only consider cost functions
$(\cost_i)_{\ipi}$ such that for all $i$, $j \in \Pi$ we have that
$\cost_i=\cost_j$ and $\cost_i:E \to \N_0$. In other words, it means
that there is a unique non-zero natural cost on every edge. Later on
we are going to compare the depths of the finite trees obtained by the
two methods.

In the case of these simple cost functions, we can directly deduce
Theorem~\ref{theo:ex EN costs} by replacing any edge of cost~$c$ by a
path of length $c$ composed of $c$ new edges (of cost 1) and then
applying the results of Section~\ref{sec:NE} on this new
game. 
If we write $\Gprime$ the new game obtained by adding new vertices and
edges when necessary, it holds that:
\begin{align*}
  |V'| & \leq |V| + (\cmax - 1)\cdot |E|\\
  & \leq |V| + (\cmax - 1)\cdot |V|^2 \text{, and}\\
  |E'| & \leq \cmax \cdot |E|\,.
\end{align*}
If we apply Proposition~\ref{prop:NE fin ds inf}, the depth~$\depth'$ of the
finite tree that is considered satisfies:
\begin{align*}
  \depth' & = (|\Pi|+1) \cdot 2 \cdot |V'|\\
  & \leq (|\Pi|+1) \cdot 2 \cdot \left(|V| + (\cmax - 1)\cdot |E|\right)\\
  & \leq (|\Pi|+1) \cdot 2 \cdot \left(|V| + (\cmax - 1)\cdot |V|^2\right)\,.
\end{align*}

Whereas if we apply Proposition~\ref{prop:NE fin ds inf costs}
directly on the initial game~$\mathcal{G}$, we have the following
equality:
$$\depth=\max\{(|\Pi| + 1)\cdot (\K+1) \cdot |V|,  
    (|\Pi|\cdot (\K+ 1)+1)\cdot |V| \cdot \K\}\,.$$

Let us first notice that if all the edges of $\mathcal{G}$ are
labelled with the same cost (i.e., $\cmax=\cmin$ and $\K=1$), then
\begin{center}
  \begin{tabular}{lll}
    $\depth'$ & $=$ & $(|\Pi|+1) \cdot 2 \cdot(|V| + (\cmax - 1)\cdot
    |E|)$, and\\ 
    $\depth$ & $=$ & $(|\Pi| + 1)\cdot 2 \cdot |V|$\,.
  \end{tabular}
\end{center}
And so,
\begin{center}
  \begin{tabular}{lll}
    if $\cmax=\cmin=1$, &
    then & $\depth' = \depth= (|\Pi| + 1)\cdot 2 \cdot |V|$, and\\
    if $\cmax=\cmin>1$, &
    then & $\depth' > \depth$\,.
  \end{tabular}
\end{center}

When $\K > 1$, the comparison between~$\depth$ and~$\depth'$ depends
on the values of many parameters of the game. For example, if the
graph of the game has five vertices, three edges of cost 1 and one
edge of cost 100, then it is more interesting to use the game
$\mathcal{G}'$ and techniques from Section~\ref{sec:NE} to construct
the Nash equilibrium, because in this case, $\depth'= (\Pi+1) \cdot 2
\cdot 104$ and $\depth= (|\Pi|\cdot 101 +1)\cdot 5 \cdot 101$, and so
$\depth >> \depth'$.

\section{Conclusion and Perspectives}
In this paper, we first prove the existence of finite-memory Nash
equilibria for quantitative multiplayer reachability games played on
finite graphs. We also prove that this result remains true when the
model is enriched by allowing $n$-tuples of non-negative costs on
edges (one cost by player), answering a question we posed
in~\cite{BBD10}. Moreover we extend our existence result to
quantitative games where both safety and reachability objectives
coexist. Secondly, we prove the existence of finite-memory secure
equilibria for quantitative two-player reachability games played on
finite graphs.

There are several interesting directions for further research. First,
we intend to investigate the existence of secure equilibria in the
$n$-player framework. Notice that the proof techniques related to our
results on secure equilibria rely on the two-player
assumption. Furthermore, we also want to investigate deeper the size
of the memory needed in the equilibria. This could be a first step
towards a study of the com\-plexi\-ty of computing equilibria with
certain requirements, in the spirit of~\cite{GU08}. We also intend to
look for existence results for \emph{subgame perfect
  equilibria}. Finally we would like to address these questions for
other objectives such as B\"uchi or request-response.

\begin{acknowledgements}
  This work has been partly supported by the ESF project GASICS and a
  grant from the National Bank of Belgium. The third author is
  supported by a grant from L'Oreal-UNESCO/F.R.S.-FNRS. The authors
  are grateful to Jean-Fran\c cois Raskin and Hugo Gimbert for useful
  discussions.
\end{acknowledgements}

\bibliographystyle{abbrv}
\bibliography{quant_games}

\end{document}